\documentclass[12pt,onecolumn,final]{IEEEtran} 
\usepackage{amsthm}

\usepackage{url}
\usepackage{epsfig,bm,booktabs,amssymb}
\usepackage[ruled]{algorithm2e}
\usepackage{mathrsfs} 
\usepackage{algorithmic}
\usepackage{subcaption}
\usepackage{amsmath, amsfonts,graphicx}
\usepackage[all]{xy}

\usepackage{chngcntr}
\counterwithout{figure}{section}
\counterwithout{table}{section}
\providecommand{\keyword}[1]{\textbf{Keywords: } #1}

\def \FIG {{Fig$.\;$}}

\newcommand{\be}{\begin{equation}}
\newcommand{\ee}{\end{equation}}
\newcommand{\ba}{\begin{eqnarray*}}
\newcommand{\ea}{\end{eqnarray*}}
\newcommand{\bi}{\begin{itemize}}
\newcommand{\ei}{\end{itemize}}

\newtheorem{thm}{\bf{Theorem}}[section]

\newtheorem{lem}{\bf{Lemma}}[section]

\newtheorem{rem}{\bf{Remark}}[section]

\newcommand{\comments}[1]{}

\graphicspath{{pic/}}
\begin{document}
\title{Minimization of Transformed $L_1$ Penalty: Closed Form Representation and 
Iterative Thresholding Algorithms}

\author{Shuai~Zhang, 
		    and~Jack~Xin
\thanks{
The work was partially supported by NSF grants DMS-0928427, DMS-1222507
and DMS-1522383. They are with the Department of Mathematics,  
University of California, Irvine, CA, 92697, USA. 
E-mail: (szhang3@uci.edu; jxin@math.uci.edu) 
Phone: (949)-824-5309. Fax: (949)-824-7993. }}


\maketitle

\begin{abstract}
The transformed $l_1$ penalty (TL1) functions are a one parameter family of
bilinear transformations composed with the absolute value function. When acting
on vectors, the TL1 penalty interpolates $l_0$ and $l_1$ similar to $l_p$ norm,  
where $p$ is in $(0,1)$. 
In our companion paper, we showed that TL1 is a robust sparsity promoting
penalty in compressed sensing (CS) problems for a broad range of incoherent 
and coherent sensing matrices. Here we develop an explicit fixed point 
representation for the TL1 regularized minimization problem. 
The TL1 thresholding functions are in closed form for all parameter values. 
In contrast, the $l_p$ thresholding functions ($p$ is in $[0,1]$) are 
in closed form only for $p=0,1,1/2,2/3$, known as hard, soft,  half, and 2/3
thresholding respectively. 
The TL1 threshold values differ in subcritical (supercritical) parameter regime 
where the TL1 threshold functions are continuous (discontinuous) similar to 
soft-thresholding (half-thresholding) functions. We propose TL1 iterative 
thresholding algorithms and compare them with hard and half thresholding 
algorithms in CS test problems. 
For both incoherent and coherent sensing matrices, a proposed TL1 iterative  
thresholding algorithm with adaptive subcritical and supercritical thresholds 
(TL1IT-s1 for short), consistently performs the best in sparse signal recovery 
with and without measurement noise. 
\end{abstract}

\keyword{
Transformed $l_1$ penalty, closed form thresholding functions, 
iterative thresholding algorithms, compressed sensing, robust recovery. }


\section{Introduction} \label{section:intro}
Iterative thresholding (IT) algorithms merit our attention in high dimensional 
settings due to their simplicity, speed and low computational costs. 
In compressed sensing (CS) problems \cite{candes2006stable,Don_06} under   
$l_p$ sparsity penalty ($p \in [0,1]$), the corresponding 
thresholding functions are in closed form 
when $p=0,\frac{1}{2},\frac{2}{3},1$. The $l_1$ algorithm is known as soft-thresholding 
\cite{soft-threshold-lp-daubechies2004iterative,Don_95}, and the $l_0$ algorithm 
hard-thresholding 
\cite{hard-sparsify-blumensath2012accelerated,hard-threshold-blumensath2008iterative}. 
IT algorithms only involve scalar thresholding and matrix multiplication. 
We note that the linearized Bregman algorithm \cite{Bregman:yin2008,YO_13} is similar 
for solving the constrained $l_1$ minimization (basis 
pursuit) problem. Recently, half and $\frac{2}{3}$-thesholding algorithms have been actively studied
\cite{xian-half-cao2013fast-image,xian-half} as non-convex alternatives to 
improve on $l_1$ (convex relaxation) and $l_0$ algorithms. 

However, the non-convex $l_p$ penalties ($p \in (0,1)$) are non-Lipschitz. 
There are also some Lipschitz continuous non-convex sparse penalties, 
including the difference of $l_1$ and $l_2$ norms (DL12)  \cite{ELX,l1-l2-yinminimization,l1-l2-lou2014computing}, 
and the transformed $l_1$ (TL1) \cite{DCATL1}. When applied to CS problems, 
the difference of convex function algorithms (DCA) of 
DL12 are found to perform the best for highly coherent 
sensing matrices. In contrast, the DCAs of TL1 are the most robust 
(consistently ranked in the top among existing algorithms) for 
coherent and incoherent sensing matrices alike.

In this paper, as companion of \cite{DCATL1}, we develop robust and effective IT algorithms 
for TL1 regularized minimization with evaluation on CS test problems. 
The TL1 penalty is a one parameter family of bilinear transformations 
composed with the absolute value function. The TL1 parameter, denoted 
by letter `$a$', plays a similar role as $p$ for $l_p$ penalty. 
If `$a$' is small (large), TL1 behaves like $l_0$ ($l_1$). 
If `$a$' is near 1, TL1 is similar to $l_{1/2}$. 
However, a strikingly different phenomenon is that 
the TL1 thresholding function is in {\it closed form 
for all values of parameter `$a$'}. 
Moreover, we found subcritical and supercritical parameter 
regimes of TL1 thresholding functions with thresholds expressed in 
different formulas. The subcritical TL1 thresholding functions are continuous, 
similar to the soft-thresholding (a.k.a. shrink) function of $l_1$ (Lasso). 
The supercritical TL1 thresholding functions have 
jump discontinuities, similar to $l_{1/2}$ or $l_{2/3}$.  

Several common non-convex penalties in statistics are SCAD \cite{SCAD}, 
MCP \cite{MC+}, log penalty \cite{SparseNet,ReweightedL1}, and 
capped $l_1$ \cite{CL1}. We refer to Mazumder, Friedman and Hastie's paper 
\cite{SparseNet} for an overview. They appeared 
in the univariate regularization problem
\[ \min \limits_x \{ \ \frac{1}{2}(x-y)^2 + \lambda \; P(x)   \ \}, \]
and produced closed form thresholding formulas. TL1 is a 
smooth version of capped $l_1$ \cite{CL1}. SCAD and MCP, corresponding to quadratic 
spline functions with one and two knots, have continuous thresholding functions.  
Log penalty and capped $l_1$ have discontinuous threshold functions. 
The TL1 thresholding function is unique in that it can be either continuous or 
discontinuous depending on parameters `a' and $\lambda$. Also similar to SCAD,
TL1 satisfies unbiasedness, sparsity and continuity conditions, 
which are desirable properties for variable selection \cite{transformed-l1,SCAD}.

The solutions of TL1 regularized minimization problem satisfy a fixed point 
representation involving matrix multiplication and thresholding only. 
Direct fixed point iterative (DFA), semi-adaptive (TL1IT-s1) 
and adaptive iterative schemes (TL1IT-s2) are proposed. 
The semi-adaptive scheme (TL1IT-s1) updates the sparsity regularization 
parameter $\lambda$ based on the sparsity estimate of the solution. 
The adaptive scheme (TL1IT-s2) 
also updates the TL1 parameter `$a$', however only doing the subcritical thresholding. 

We carried out extensive sparse signal recovery experiments in section \ref{section:experiment},  
with three algorithms: TL1IT-s1, Hard and Half-thresholding methods. 
For Gaussian sensing matrices with positive covariance, 
TL1IT-s1 leads the pack and half-thresholding is the second. 
For coherent over-sampled discrete cosine transform (DCT) matrices, 
TL1IT-s1 is again the leader and with considerable margin. 
The half thresholding algorithm drops to the distinct last.
In the presence of measurement noise, the results are similar, with 
TL1IT-s1 maintaining its leader status in both classes of random sensing 
matrices. That TL1IT-s1 fairs much better than other methods 
may be attributed to the two built-in thresholding values. The early iterations 
are observed to go between the subcritical and supercritical regimes frequently. 
Also TL1IT-s1 is stable and robust when exact sparsity of solution is 
replaced by rough estimates as long as the number of linear 
measurements exceeds a certain level.  

The rest of the paper is organized as follows. In section \ref{section: theory}, we 
overview TL1 minimization. 
In section \ref{TL1rep}, we derive TL1 thresholding functions in 
closed form and show their continuity properties with details of the proof left in the appendix. 
The analysis is elementary yet delicate, and makes use of 
the Cardano formula on roots of cubic polynomials and algebraic identities. The 
fixed point representation for the TL1 regularized optimal solution follows.
In section \ref{section:algorithm}, we propose three TL1IT schemes and derive the 
parameter update formulas for TL1IT-s1 and TL1IT-s2 based on the thresholding functions. 
We analyze convergence of the fixed parameter TL1IT algorithm.
In section \ref{section:experiment}, numerical experiments on CS test problems are carried out 
for TL1IT-s1, hard and half thresholding algorithms on Gaussian and over-sampled DCT 
matrices with a broad range of coherence. The TL1IT-s1 leads 
in all cases, and inherits well the robustness and effective sparsity 
promoting capability of TL1 \cite{DCATL1}.
Concluding remarks are in section \ref{section: conclusion}. 

\section{Overview of TL1 Minimization} \label{section: theory}
\setcounter{equation}{0}
The transformed $l_1$ (TL1) function $\rho_a(x)$ is defined as 
\begin{equation}\label{TL1a}
   \rho_a(x)  =  \dfrac{(a+1)|x|}{a+|x|} \ ,  
\end{equation}
where parameter $a \in (0,+\infty)$; see \cite{transformed-l1} for its 
unbiasedness, sparsity and continuity properties. 
With the change of parameter `a',  TL1 interpolates $l_0$ and $l_1$ norms:   
\[ \lim_{a \to 0^{+}} \rho_a(x) = I_{\{x \neq 0\}} , \  \
 \lim_{a \to +\infty} \rho_a(x) = |x|. \]
In Fig$.\;$\ref{gp: level lines}, level lines of TL1 on the plane are shown   
at small and large values of parameter $a$, resembling  
those of $l_1$ (at $a = 100$), $l_{1/2}$ (at $a=1$), and 
$l_0$ (at $a=0.01$). 

\begin{figure}
\def\arraystretch{3}
\begin{tabular}{l r}
\begin{minipage}[t]{0.45\linewidth}
\includegraphics[scale=0.35]{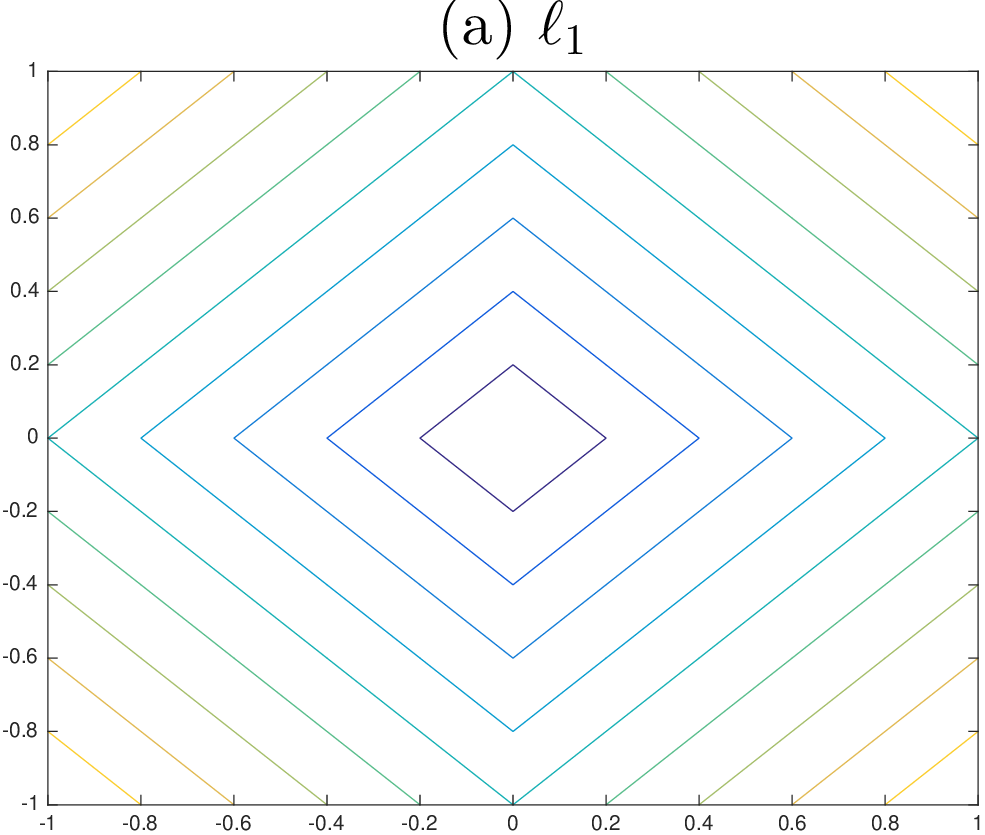} 
\end{minipage} & 
\begin{minipage}[t]{0.45\linewidth}
\includegraphics[scale= 0.35]{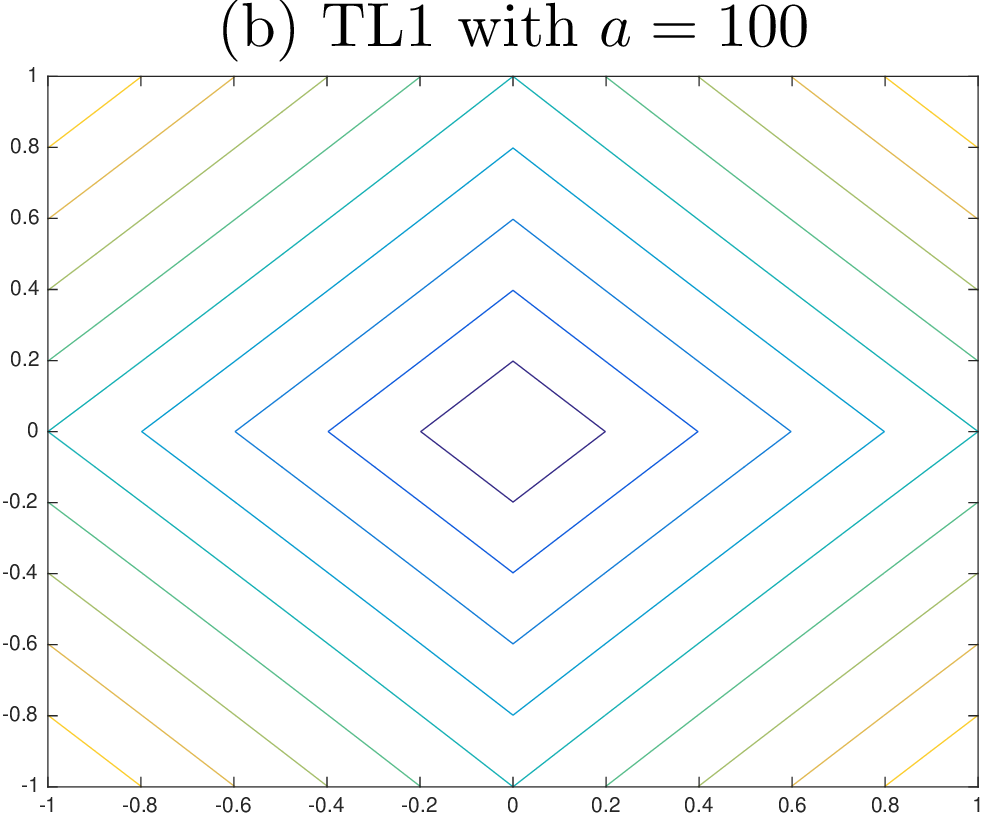} 
\end{minipage} \\
\begin{minipage}[t]{0.45\linewidth}
\includegraphics[scale=0.35]{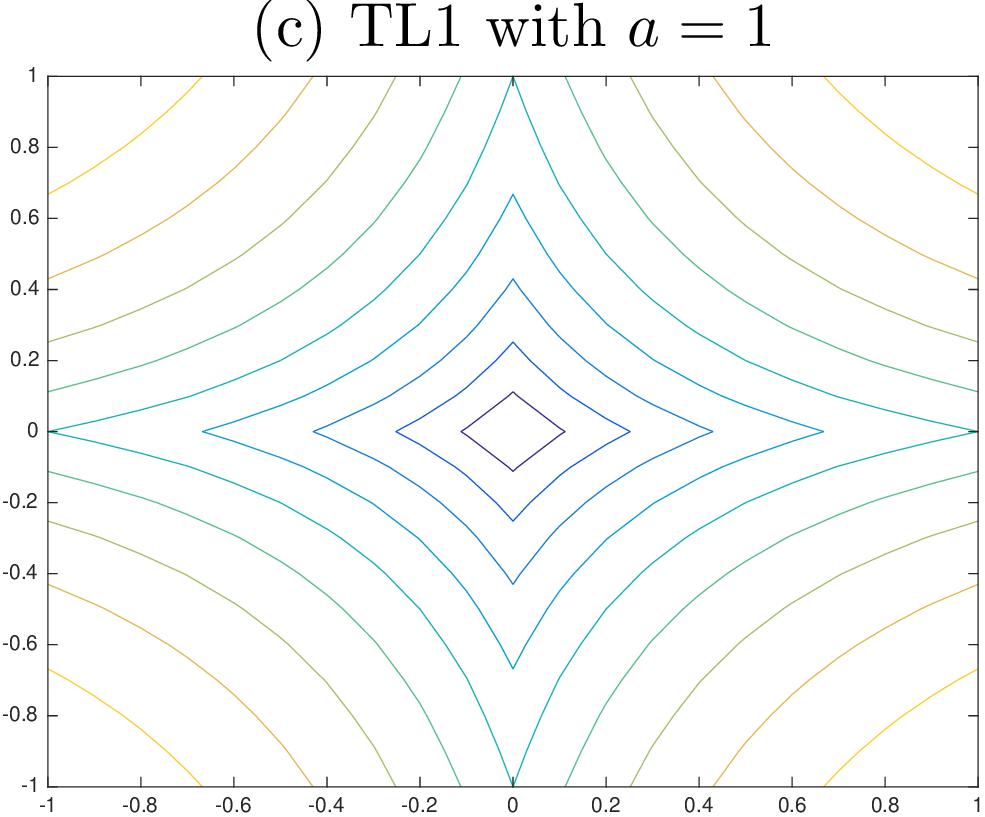} 
\end{minipage} & 
\begin{minipage}[t]{0.45\linewidth}
\includegraphics[scale= 0.35]{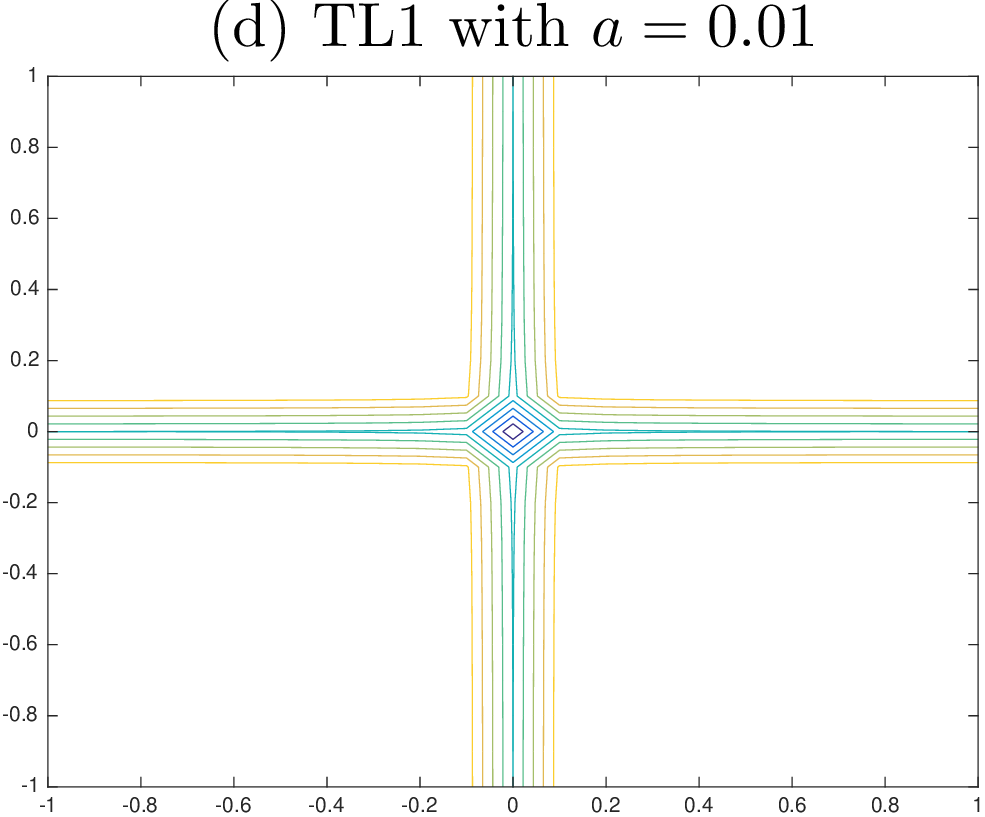} 
\end{minipage}
\end{tabular} 
\caption{Level lines of TL1 with different parameters: $a=100$ (figure b), $a=1$ (figure c), 
$a=0.01$ (figure d). For large parameter a, the graph looks almost the same as $l_1$ (figure a). 
While for small value of a, it tends to the axis.}
\label{gp: level lines}
\end{figure}

Next, we want to expand the definition of TL1 to vector space. For vector 
$x = (x_1, x_2, \cdots , x_N)^T \in \Re^N$, we define
\begin{equation}\label{TL1b}
P_a(x) = \sum \limits_{i = 1}^N \rho_a(x_i).
\end{equation} 

In this paper, we will use TL1 instead of $l_0$ norm to solve application problems 
proposed from compressed sensing. The mathematical models can be generalized 
as two categories: 
the constrained TL1 minimization:   
\begin{equation} \label{model: cons-optim}
   \min \limits_{x \in \Re^N} f(x) =  \min \limits_{x \in \Re^N} P_a(x)  \;\; s.t. \;\;\  Ax=y, 
\end{equation}
and the unconstrained TL1-regularized minimization:
\begin{equation} \label{model: uncons-optim}
   \min \limits_{x \in \Re^N} f(x) = \min \limits_{x \in \Re^N} \frac{1}{2} \|Ax - y \|^2_2 + \lambda P_a(x),
\end{equation}
where $\lambda$ is the trade-off Lagrange multiplier to control the amount of shrinkage.

The exact and stable recovery by TL1 for (\ref{model: cons-optim})
under the Restricted Isometry Property (RIP) \cite{candes2005decoding, candes2006stable}
conditions is established in the companion paper \cite{DCATL1}, where the 
difference of convex functions algorithms (DCA) 
for (\ref{model: cons-optim}) and (\ref{model: uncons-optim}) are also presented and compared with 
some state-of-the-art CS algorithms on sparse signal recovery problems.
In paper \cite{DCATL1},  the authors find that TL1 is always among top performers 
in RIP and non-RIP categories alike. 
However, matrix multiplication and inverse operations are 
involved at each iteration step of TL1 DC algorithms, which increases run time and 
computation costs. Iterative thresholding (IT) algorithms usually are much faster, 
since only matrix-vector multiplications and elementwise scalar thresholding operations are needed. 
Also, due to precise threshold values, it needs fewer steps in IT to 
converge to sparse solutions.    
In order to reduce computation time, we shall explore thresholding property for TL1 penalty. 
In another paper \cite{TS1}, we expand 
TL1 thresholding and representation theories to low rank matrix completion 
problems via Schatten-1 quasi-norm. 
 
\section{Thresholding Representation and Closed-Form Solutions}\label{TL1rep}
\setcounter{equation}{0}
The thresholding theories and algorithms for $l_0$ quasi-norm (hard-thresholding) 
\cite{hard-sparsify-blumensath2012accelerated,hard-threshold-blumensath2008iterative} 
and $l_1$ norm (soft-thresholding) \cite{soft-threshold-lp-daubechies2004iterative,Don_95} 
are well-known and widely tested. Recently, the closed form thresholding representation theories 
and algorithms for $l_{p}$ ($p=1/2,2/3$) regularized problems are proposed 
\cite{xian-half-cao2013fast-image,xian-half} based on Cardano's root formula of cubic polynomials. 
However, these algorithms are limited to few specific values of parameter $p$. 
Here for TL1 regularization problem, we derive the closed form representation of optimal solution, under
{\it any positive value of parameter} $a$. 

Let us consider the unconstrained TL1 regularization model (\ref{model: uncons-optim}): 
\begin{equation*}
\min \limits_{x} \frac{1}{2} \| Ax -y \|_2^2 + \lambda P_a(x),
\end{equation*}
for which the first order optimality condition is: 
\begin{equation}\label{equ: diff of P}
0 = A^T(Ax-y) + \lambda \cdot \nabla P_a(x).
\end{equation}
Here $\nabla P_a(x) = \left( \, \partial \rho_a(x_1), 
\, ...\, , \partial \rho_a(x_N) \, \right)$, and 
$\partial \rho_a(x_i) = \dfrac{a(a+1) SGN(x_i)}{(a+|x_i|)^2}$. 
$SGN(\cdot)$ is the set-valued signum function with $SGN(0) \in [-1,1]$, instead of a single fixed value. In this paper, we will use $sgn(\cdot)$ to represent the standard signum function with $sgn(0)=0$.
From equation (\ref{equ: diff of P}), it is easy to get 
\begin{equation} \label{equ: optimal condition}
x + \mu A^T(y-Ax) = x + \lambda \mu \nabla P_a(x).
\end{equation}

We can rewrite the above equation, via introducing two operators  
\begin{equation} \label{func: R&B}
\begin{array}{l}
\vspace{1mm}
R_{\lambda \mu, a}(x) = [I + \lambda \mu \nabla P_a(\cdot)]^{-1}(x), \\
B_{\mu}(x) = x + \mu A^T(y- Ax).
\end{array}
\end{equation}
From equation (\ref{equ: optimal condition}), we will get a representation
equation for optimal solution $x$: 
\begin{equation} \label{equ: repre nessa mini}
x = R_{\lambda \mu, a}(B_{\mu}(x)).
\end{equation}
We will prove that the operator $R_{\lambda \mu, a}$ is diagonal under some requirements for parameters $\lambda$, $\mu$ and $a$. Before that, 
a closed form expression of proximal operator at scalar TL1 $\rho_a(\cdot)$ 
will be given and proved at following subsection.  This optimal solution 
expression will be used to prove the threshold representation theorem for 
model (\ref{model: uncons-optim}).


\subsection{Proximal Point Operator for TL1}
Like \cite{Proximal}, we introduce proximal operator 
$prox_{\lambda \rho_a}: \Re \rightarrow \Re$ 
for univariate TL1 ($\rho_a$) regularization problem, 
\[
prox_{\lambda \rho_a}(y) = arg \min \limits_{x \in \Re} \left( \frac{1}{2}(y-x)^2 + \lambda \rho_a(y) \right).
\]
Proximal operator of a convex function usually intends to solve a small convex 
regularization problem, which often admits closed-form formula or an efficient
specialized numerical methods. However, for non-convex functions, like $l_p$
with $p \in (0.1)$, their related proximal operators do not have closed form 
solutions in general. There are many iterative algorithms to approximate 
optimal solution. But they need more computing time and sometimes only converge to local optimal or stationary point. 
In this subsection, we prove that for TL1 function, there indeed exists a 
closed-formed formula for its optimal solution.

For the convenience of our following theorems, we want to introduce three parameters: 
\begin{equation} \label{form: threshold parameters}
\left\lbrace \begin{array}{l}
   \vspace{2mm}
   t^*_1 = \dfrac{3}{2^{2/3}} (\lambda a(a+1))^{1/3} -a \\
   t^*_2 = \lambda \frac{a+1}{a} 
   \vspace{2 mm} \\
   t^*_3 = \sqrt{2\lambda (a+1)} - \frac{a}{2}. \\
\end{array} \right.
\end{equation}
It can be checked that inequality $t^*_1 \leq t^*_3 \leq t^*_2$ holds. 
The equality is realized if  
$\lambda = \frac{a^2}{2(a+1)}$ (Appendix A).

\begin{lem} \label{lem: roots for poly}
For different values of scalar variable $x$, the roots of the following two 
cubic polynomials in $y$ satisfy properties:
\begin{enumerate}
\item If $x > t^*_1$,
there are 3 distinct real roots of the cubic polynomial:
\begin{equation*}
y(a+y)^2 - x(a+y)^2 + \lambda a(a+1) = 0.
\end{equation*}
Furthermore, the largest root $y_0$ is given by $y_0 = g_{\lambda}(x)$, where 
\begin{equation}  \label{func: g formula}
g_{\lambda}(x) = sgn(x) \left\{ \frac{2}{3}(a+|x|)cos(\frac{\varphi(x)}{3}) 
                                 -\frac{2a}{3} + \frac{|x|}{3} \right\}
\end{equation}
with $ \varphi(x) = \arccos( 1 - \frac{27\lambda a(a+1)}{2(a+|x|)^3} ) $, 
and $|g_{\lambda}(x)| \leq |x|$.

\item If $x < -t^*_1$, 
there are also 3 distinct real roots of cubic polynomial:
\begin{equation*}
y(a-y)^2 - x(a-y)^2 - \lambda a(a+1) = 0.
\end{equation*}
Furthermore, the smallest root denoted by $y_0$, is given by $y_0 = g_{\lambda}(x)$.
\end{enumerate}
\end{lem}

\begin{proof}
\begin{enumerate}
\item[1.)] 
First, we consider the roots of cubic equation: 
\[
	y(a+y)^2 - x(a+y)^2 + \lambda a(a+1) = 0 \text{, when $x>t^*_1$.}
\]

We apply variable substitution 
$\eta = y + a$ 
in the above equation, then it becomes
\[ 
	\eta^3 - (a+x)\eta^2 + \lambda a (a+1) = 0,
\]
whose discriminant is: 
\[ 
	\bigtriangleup = \lambda\, (a+1)\, a \, [4(a+x)^3 - 27\lambda\, (a+1)\, a ].
\]
Since $ x \geq t^* $ and $\bigtriangleup > 0$, 
there are three distinct real roots for this cubic equation. 

Next, we change variables as $\eta = t + \frac{a}{3} + \frac{x}{3} = y + a$. 
The relation between $y$ and $t$ is: $ y = t - \frac{2a}{3} + \frac{x}{3} $. 
In terms of $t$, the cubic polynomial is turned into a depressed cubic as: 
\[ t^3 + pt + q = 0, \] 
where $p = -(a+x)^2/3$, and $q = \lambda a(a+1) - 2(a+x)^3/27$.
The three roots in trigonometric form are: 
\begin{equation} \label{equ: three roots}
\begin{array}{l}
t_0 = \frac{2(a+x)}{3} \, \cos(\varphi /3) \\
t_1 = \frac{2}{3} (a+x) \, \cos(\varphi/3 + \pi/3) \\
t_2 = -\frac{2}{3}(a+x)\, \cos(\pi/3 - \varphi/3) \\
\end{array}
\end{equation}
where $\varphi = \arccos(1 - \frac{27\lambda a(a+1)}{2(a+x)^3}) $.

Then $t_2 < 0$, and 
$t_0 > t_1 > t_2.$ 
By the relation 
$y = t - \frac{2a}{3} + \frac{x}{3}$, 
the three roots in variable $y$ are: 
$y_i = t_i - \frac{2a}{3} + \frac{x}{3}$, for $i = 1,2,3$. 
From these formula, we know that: 
\[
	y_0 > y_1 > y_2.
\] 
Also it is easy to check that 
$y_0 \leq x$ and $y_2 < 0$, and the largest root $y_0 = g_{\lambda}(x)$, 
when $x>t^*_1$.  

\item[2.)] 
Next, we discuss the roots of the cubic equation: 
\[ 
	(a-y)^2y - x(a-y)^2 - \lambda a(a+1) = 0 \text{, when $x<-t^*_1$}.
\]
Here we set: $\eta = a - y$, and $t = \eta + \frac{x}{3} - \frac{a}{3}$. 
So $y = -t + \frac{x}{3} + \frac{2a}{3}$. 
By a similar analysis as in part (1), there are 3 distinct roots for polynomial 
equation: $y_0 < y_1 < y_2$ with the smallest solution
\[
	y_0 = -\frac{2}{3}(a-x)\, \cos(\varphi/3) + \frac{x}{3} + \frac{2a}{3},
\] 
where $\varphi = \arccos(1 - \frac{27\lambda a(a+1)}{2(a-x)^3})$.
So we proved that the smallest solution is 
$y_0 = g_{\lambda}(x)$, when $x < -t^*_1$.
\end{enumerate}
\end{proof}

Next let us define the function $f_{\lambda,x}(\cdot): \Re \rightarrow \Re$,
\begin{equation}
   f_{\lambda,x}(y) = \frac{1}{2}(y-x)^2 + \lambda \rho_a(y).
\end{equation}
So $\partial f_{\lambda,x}(y) = y - x + \lambda \frac{a(a+1)SGN(y)}{(a+|y|)^2}$. \\

\begin{thm} \label{theorem: g formula}
The optimal solution $y_\lambda^*(x) = arg \min \limits_{y} f_{\lambda,x}(y)$ 
is a threshold function with threshold value $t$ :
\begin{equation} \label{TL1thrfunc} 
y_\lambda^*(x) = 
\left\{
	\begin{array}{ll}
		0 ,             & |x| \leq t \\
		g_{\lambda}(x), & |x| > t
	\end{array} 
\right.
\end{equation}
where $g_{\lambda}(\cdot)$ is defined in (\ref{func: g formula}). 
The threshold parameter $t$ depends on regularization parameter $\lambda$,  
\begin{enumerate}
\item  if $\lambda \leq \frac{a^2}{2(a+1)}$ (sub-critical),
$$t = t^*_2 = \lambda \frac{a+1}{a};$$
\item $\lambda > \frac{a^2}{2(a+1)}$ (super-critical),
$$t = t^*_3 =  \sqrt{2\lambda (a+1)} - \frac{a}{2}, $$
\end{enumerate}
where parameters $t_2^*$ and $t_3^*$ are defined 
 in formula (\ref{form: threshold parameters}).
\end{thm}

\begin{proof}
In the following proof, we represent $y_\lambda^*(x)$ as $y^*$  for simplicity. 
We split the value of $x$ into 3 cases: $x=0$, $x>0$ and $x<0$, then prove 
our conclusion case by case. 
\begin{enumerate}
\item[1.)] $x=0$. 

In this case, optimization objective function is 
$f_{\lambda,x}(y) = \frac{1}{2} y^2 + \lambda \rho_a(y)$. 
Here the two factors $\frac{1}{2} y^2$ and $\lambda \rho_a(|y|)$ 
are both increasing for $y>0$, and decreasing for $y<0$. 
Thus $f(0)$ is the unique minimizer for function $f_{\lambda,x}(y)$. 
So 
\[ 
	y^* = 0  \text{, when } x = 0.
\]

\item[2.)] $x > 0$.

Since $\frac{1}{2} (y-x)^2$ and $\lambda \rho_a(y)$ 
are both decreasing for $y<0$, 
our optimal solution will only be obtained  at nonnegative values.
Thus it just needs to consider all positive stationary points for function 
$f_{\lambda}(y)$ and also point $0$. 

When $y > 0$, we have:
\[ 
	f^{'}_{\lambda,x} (y) = y - x + \lambda \dfrac{a(a+1)}{(a+y)^2}, 
\]
and 
\[ 
	f^{''}_{\lambda,x} (y) = 1 - 2 \lambda \dfrac{a(a+1)}{(a+y)^3}.
\]
	
Since $f_{\lambda,x}^{''}(y)$ is increasing,  
$f_{\lambda,x}^{''}(0) = 2 \lambda \frac{(a+1)}{a^2}$ 
determines the convexity for the function $f(y)$. In the following proof, we
further discuss the value of $y^*$ by two conditions: 
$\lambda \leq \frac{a^2}{2(a+1)}$ and $\lambda > \frac{a^2}{2(a+1)}$.

\begin{enumerate}
\item[2.1)] $\lambda \leq \frac{a^2}{2(a+1)}$.

So we have 
$\inf \limits_{y > 0} f^{''}_{\lambda}(y) = f^{''}_{\lambda}(0+) 
	= 1 - 2 \lambda \frac{(a+1)}{a^2} \geq 0,$ 
which means function $f^{'}_{\lambda}(y)$ is increasing for $y \geq 0$, 
with minimum value 
$f^{'}_{\lambda}(0) = \lambda \frac{(a+1)}{a} - x = t^*_2 - x$.

$i) $ When $0 \leq x \leq t_2^*$, 
$f^{'}_{\lambda,x}(y)$ is always positive, 
thus the optimal value $y^* = 0.$ 

$ii) $ When $x > t_2^*$,  
$f^{'}_{\lambda,x}(y)$ is first negative then positive. 
Also $x \geq t_2^* \geq t_1^*$. 
The unique positive stationary point $y^*$ of $f_{\lambda,x}(y)$ 
satisfies equation: 
$f^{'}_{\lambda}(y^*) = 0,$ which implies  
\begin{equation} \label{equ: cubic in theorem 1}
	y(a+y)^2 - x(a+y)^2 + \lambda a(a+1) = 0.
\end{equation} 
According to Lemma \ref{lem: roots for poly}, 
the optimal value $y^* = y_0 =  g_{\lambda}(x)$. 

Above all, the value for $y^*$ is : 
\begin{equation}
y^* = \left\{ 
\begin{array}{ll}
0 , & 0 \leq x \leq t_2^*; \\
g_{\lambda}(x), & x > t_2^*
\end{array} \right.
\end{equation}
under the condition $\lambda \leq \frac{a^2}{2(a+1)}$.

\item[2.2)] $\lambda > \frac{a^2}{2(a+1)}$.

In this case, due to the sign of $f^{''}_{\lambda}(y)$, we know that function
 $f^{'}_{\lambda,x}(y)$ is decreasing at first then switches to be increasing at 
 the domain $[0,\infty)$. Its minimum obtained at point 
$\overline{y}  = (2\lambda a (a+1))^{1/3} - a$ 
and 
\[
f^{'}_{\lambda}(\overline{y}) 
	= \frac{3}{2^{2/3}}(\lambda (a+1)a)^{1/3} -a - x  
	= t_1 - x.
\] 
Thus $ f^{'}_{\lambda}(y) \geq t_1^* - x$, for $y \geq 0$. 

$i) $ When $0 \leq x \leq t_1^*$, 
function $f_{\lambda}(y)$ is always increasing. Thus optimal value $y^* = 0$.

$ii) $ When $ t_2^* \leq x$, $f^{'}_{\lambda}(0+) \leq 0$. 
So function $f_{\lambda}(y)$ is decreasing first, then increasing. 
There is only one positive stationary point, which is also the optimal solution. 
Using Lemma \ref{lem: roots for poly},
we know that $y^* = g_{\lambda}(x)$.

$iii) $ When $ t_1^* < x < t_2^*$, $f^{'}_{\lambda}(0+) > 0$. 
Thus function $f_{\lambda}(y)$ is first increasing, then 
decreasing and finally increasing, which implies that 
there are two positive stationary points and the larger one is a local minima. 
Using Lemma \ref{lem: roots for poly} again, 
the local minimize point will be $y_0 = g_{\lambda}(x)$, the largest root 
of equation (\ref{equ: cubic in theorem 1}). But we still need to compare 
$f_{\lambda}(0)$ and $f_{\lambda}(y_0)$ to distinguish the global optimal $y^*$.
Since $ y_0 - x + \lambda \frac{a(a+1)}{(a+y_0)^2} = 0$, which implies 
$\lambda \frac{(a+1)}{a+y_0} = \frac{(x-y_0)(a+y_0)}{a}$,
we have  
\begin{equation}
\begin{array}{ll}
f_{\lambda}(y_0) - f_{\lambda}(0) 
              & = \frac{1}{2}y_0^2 - y_0 x + \lambda  \frac{(a+1)y_0}{a+y_0} \\
              & = y_0 ( \frac{1}{2}y_0 -  x + \lambda  \frac{(a+1)}{a+y_0})  \\
              & = y_0 ( \frac{1}{2}y_0 -  x + \frac{(x-y_0)(a+y_0)}{a}) \\
              & = y_0^2 ( \frac{x-y_0}{a} - \frac{1}{2} ) 
= y_0^2 ( (x- g_{\lambda} (x))/a - 1/2 )  
\end{array}
\end{equation}
It can be proved that parameter $t^*_3$ is the unique root of 
$t - g_{\lambda}(t) - \frac{a}{2} = 0$ in $[t^*_1, t^*_2]$
(see Appendix B). 
For $ t^*_1 \leq t \leq t^*_3$, $t - g_{\lambda}(t) - \frac{a}{2} \geq 0$;
for $ t^*_3 \leq t \leq t^*_2$, $t - g_{\lambda}(t) - \frac{a}{2} \leq 0$.
So in the third case: $ t_1^* < x < t_2^*$: 
if $t_1^* < x \leq t^*_3$, $y^* = 0$; 
if $x > t^*_3$,   $y^* = y_0= g_{\lambda}(x).$

Finally we know that under the condition $\lambda > \frac{a^2}{2(a+1)} :$ 
\begin{equation}
y^* = \left\{ 
\begin{array}{ll}
0 ,    & 0 \leq x \leq t^*_3; \\
g_{\lambda}(x),   & x > t^*_3,
\end{array} \right.
\end{equation}

\end{enumerate}

\item[3.)] $x < 0$.

Notice that 
\[
	\inf \limits_{y} \, f_{\lambda,x}(y) 
		= \inf \limits_{y} \, f_{\lambda,x}(-y) 
		= \inf \limits_{y} \  \, \frac{1}{2}(y - |x|))^2 + \rho_a (y) ,
\] 
so $y^*(x) = -y^*(-x)$, which implies that 
the formula obtained when $x > 0$ above, can extend 
to the case: $x < 0$ by odd symmetry. Formula (\ref{TL1thrfunc}) holds.

%
\end{enumerate}

Summarizing results from all cases, the proof is complete. 

\end{proof}

%
%
%

\subsection{Optimal Point Representation for Regularized TL1 
(\ref{model: uncons-optim})}

Next, we will show that the optimal solution of the TL1 regularized
problem (\ref{model: uncons-optim})  can be expressed by a thresholding function. 
Let us introduce two auxiliary objective functions. 
For any given positive parameters $\lambda$, $\mu$ and 
vector $z \in \Re^N$, define:  
\begin{equation} \label{func surrogate}
\begin{array}{l}
C_{\lambda}(x) = \frac{1}{2}\|y - Ax\|_2^2 + \lambda P_a(x)  \\
C_{\mu}(x,z) = \mu \left\{ C_{\lambda}(x) - \frac{1}{2}\|Ax-Az\|_2^2 \right\} 
               + \frac{1}{2}\|x-z\|_{2}^{2}.
\end{array}
\end{equation}
The first function $C_{\lambda}(x)$ comes from the objective of TL1 regularization problem (\ref{model: uncons-optim}).

Starting from this subsection till the end of this paper, we substitute parameter 
$\lambda$ in threshold value $t_i^*$ with the product of $\lambda$ and $\mu$,
which are
\begin{equation}
\left\lbrace \begin{array}{l}
   \vspace{2mm}
   t^*_1 = \dfrac{3}{2^{2/3}} (\lambda \mu a(a+1))^{1/3} -a \\
   t^*_2 = \lambda \mu \frac{a+1}{a} 
   \vspace{2 mm} \\
   t^*_3 = \sqrt{2\lambda \mu (a+1)} - \frac{a}{2}. \\
\end{array} \right.
\end{equation}

\begin{lem}  \label{lem 1}
If $x^s = (x^s_1, \cdots ,x^s_N)^T$ is a minimizer of $C_{\mu}(x,z)$ 
with fixed parameters $\{ \mu, a, \lambda, z\}$,
then there exists a positive number 
$t = t^*_2 \, I_{\left\{ \lambda \mu \leq \frac{a^2}{2(a+1)} \right\}} 
		+ t^*_3\, I_{\left\{ \lambda \mu > \frac{a^2}{2(a+1)} \right\}} $, 
such that: for $i = 1, \cdots , N$, 
\begin{equation}
\begin{array}{ll}
x_i^s =0,                                 
	& \text{ when } abs([B_{\mu}(z)]_i) \leq t; \\
x_i^s = g_{\lambda \mu }([B_{\mu}(z)]_i), 
	& \text{ when } abs([B_{\mu}(z)]_i) > t.
\end{array}
\end{equation}
Here the function $g_{\lambda \mu }(\cdot)$ is same as (\ref{func: g formula}) 
with parameter $\lambda \mu$ in place of $\lambda$ there.
$B_{\mu}(z) = z + \mu A^T(y- Az) \in \Re^N$, as in (\ref{func: R&B}).
\end{lem}

\begin{proof} The second auxiliary objective function
can be rewritten as
\begin{equation} \label{equ: C x n+1 n relation}
\begin{array}{rll}
C_{\mu}(x,z) & = & \frac{1}{2} \|x-[(I-\mu A^TA)z + \mu A^Ty]\|_2^2 + \lambda \mu P_a(x)  \\ 
             \vspace{2mm}
             &   & + \frac{1}{2}\mu\|y\|_2^2 + \frac{1}{2}\|z\|_2^2 - \frac{1}{2}\mu \|Az\|^2_2 
                   - \frac{1}{2} \|(I - \mu A^TA)z+\mu A^Ty\|_2^2 \\
                   
             & = & \frac{1}{2} \sum \limits_{i = 1}^{N} (x_i-[B_{\mu}(z)]_i)^2 
                   + \lambda \mu \sum \limits_{i = 1}^{N} \rho_a(x_i)  \\
             &   & + \frac{1}{2}\mu\|y\|_2^2 + \frac{1}{2}\|z\|_2^2 - \frac{1}{2}\mu \|Az\|^2_2  
                   - \frac{1}{2} \|(I - \mu A^TA)z+\mu A^Ty\|_2^2 ,
\end{array}
\end{equation}
which implies that
\begin{equation} \label{equa: optimal proximal}
\begin{array}{lll}
\vspace{2mm}
x^s 
	& = & arg \min \limits_{x \in \Re^N}  C_{\mu}(x,z) \\
   	& = & arg \min \limits_{x \in \Re^N} \left\{ 
    		\frac{1}{2} \sum \limits_{i = 1}^{N} (x_i-[B_{\mu}(z)]_i)^2                   
         + \lambda \mu \sum \limits_{i = 1}^{N} \rho_a(x_i)  \right\}
\end{array}
\end{equation}

Since each component $x_i$ is decoupled, the above minimum can be 
calculated by minimizing with respect to each $x_i$ individually. 
For the component-wise minimization, the objective function is : 
\begin{equation}
f(x_i,z) = \frac{1}{2} (x_i-[B_{\mu}(z)]_i)^2 + \lambda \mu \rho_a(|x_i|).
\end{equation}

Then by Theorem (\ref{theorem: g formula}), the proof of our Lemma is complete.

\end{proof}


Based on Lemma \ref{lem 1}, we have the following representation theorem.
\begin{thm}  \label{theorem 2}
If $x^* = (x^*_1,x_2^*,...,x^*_N)^T$ is a TL1 regularized solution 
of (\ref{model: uncons-optim}) with $a$ and $\lambda$ being positive constants, 
and $0 < \mu < \|A\|^{-2}$, then letting
$t = t^*_2 1_{\left\{ \lambda \mu \leq \frac{a^2}{2(a+1)} \right\}} 
     + t^*_3 1_{\left\{ \lambda \mu > \frac{a^2}{2(a+1)} \right\}} $, 
the optimal solution satisfies
\begin{equation}
x_i^* = \left\{ 
\begin{array}{ll}
g_{\lambda \mu }([B_{\mu}(x^*)]_i) , 
	& \text{ if }  |[B_{\mu}(x^*)]_i| > t \\
0, 
    & others.
\end{array}
\right.
\end{equation}
\end{thm}

\begin{proof} 
The condition $0 < \mu < \|A\|^{-2}$ implies
\begin{equation}
\begin{array}{lll}
C_{\mu}(x,x^*) & = &  \mu\{  \frac{1}{2}\|y - Ax\|_2^2 + \lambda P_a(x) \} \\ 
               &   &  + \frac{1}{2} \{ - \mu \|Ax-Ax^*\|_2^2 + \|x-x^*\|_2^2 \} \\
            & \geq &  \mu\{  \frac{1}{2}\|y - Ax\|_2^2 + \lambda P_a(x) \} \\
            & \geq &  C_{\mu}(x^*,x^*),
\end{array}
\end{equation}
for any $x \in \Re^N$. So it shows that $x^*$ is a minimizer of $C_{\mu}(x,x^*)$ 
as long as $x^*$ is a TL1 solution of (\ref{model: uncons-optim}). 
In view of Lemma (\ref{lem 1}), we finish the proof.

\end{proof}


%



\section{TL1 Thresholding Algorithms} \label{section:algorithm}
\setcounter{equation}{0}
In this section, we propose 3 iterative thresholding algorithms for regularized 
TL1 optimization problem (\ref{model: uncons-optim}), based on Theorem 
\ref{theorem 2}. 

We want to introduce a thresholding operator $G_{\lambda \mu, a}(\cdot): \Re \rightarrow \Re $ as
\begin{equation} \label{alg: threshold operator}
G_{\lambda \mu, a}(w) = \left\{
\begin{array}{ll}
0, & \quad  \text{if} \ \ |w| \leq t; \\
g_{\lambda \mu}(w), & \quad \text{if} \ \ |w| > t.
\end{array} \right.
\end{equation}
and expand it to vector space $\Re^N$,  
\[G_{\lambda \mu,a}(x) = \left( G_{\lambda \mu, a}(x_1), ... , G_{\lambda \mu, a}(x_N)  \right). \]
According to Theorem \ref{theorem 2}, optimal solution of model 
(\ref{model: uncons-optim}) satisfies representation equation 
\begin{equation} \label{equa: fixed representation}
x = G_{\lambda \mu, a}(B_{\mu}(x)).
\end{equation} 

\subsection{Direct Fixed Point Iterative Algorithm --- DFA }
A natural idea is to develop an iterative algorithm based on the above fixed point representation directly, with fixed 
values for parameters: 
$\lambda, \mu$ and $a$. 
We call it direct fixed point iterative algorithm (DFA), 
for which the iterative scheme is 
\begin{equation} \label{alg: DFA}
x^{n+1} = G_{\lambda \mu, a}(x^n + \mu A^T(y - Ax^n)) = G_{\lambda \mu, a}(B_{\mu}(x^n)),
\end{equation}
at $(n+1)$-th step.  
Recall that the thresholding parameter $t$ is: 
\begin{equation} \label{form: threshold lambda}
t= \left\{
\begin{array}{ll}
t^*_2 = \lambda \mu \frac{a+1}{a}, & \quad \text{if} \ \ \lambda \leq \frac{a^2}{2(a+1) \mu}, \\
t^*_3 = \sqrt{2\lambda \mu (a+1)} - \frac{a}{2},   & \quad \text{if} \ \ \lambda > \frac{a^2}{2(a+1) \mu}.
\end{array} \right.
\end{equation}

In DFA, we have 2 tuning parameters: product term $\lambda \mu$ and
TL1 parameter $a$, which are fixed and can be determined by cross-validation based on different categories of matrix $A$. 
Two adaptive iterative thresholding (IT) algorithms will be introduced later.  

\begin{rem}
In TL1 proximal thresholding operator $G_{\lambda \mu, a}$, the threshold value $t$ 
varies with other parameters:
\[
	t = t^*_2 \, I_{\left\{ \lambda \mu \leq \frac{a^2}{2(a+1)} \right\}} 
    		+ t^*_3 \, I_{\left\{ \lambda \mu > \frac{a^2}{2(a+1)} \right\}}.
\]
Since $t\geq t^*_3 = \sqrt{2\lambda \mu (a+1)} - \frac{a}{2}$, the larger 
the $\lambda$, the larger the threshold value $t$, and therefore the sparser the solution 
from the thresholding algorithm.
\end{rem}

\begin{figure}
\begin{tabular}{ll}
\begin{minipage}[t]{0.45\linewidth}
\includegraphics[scale=0.35]{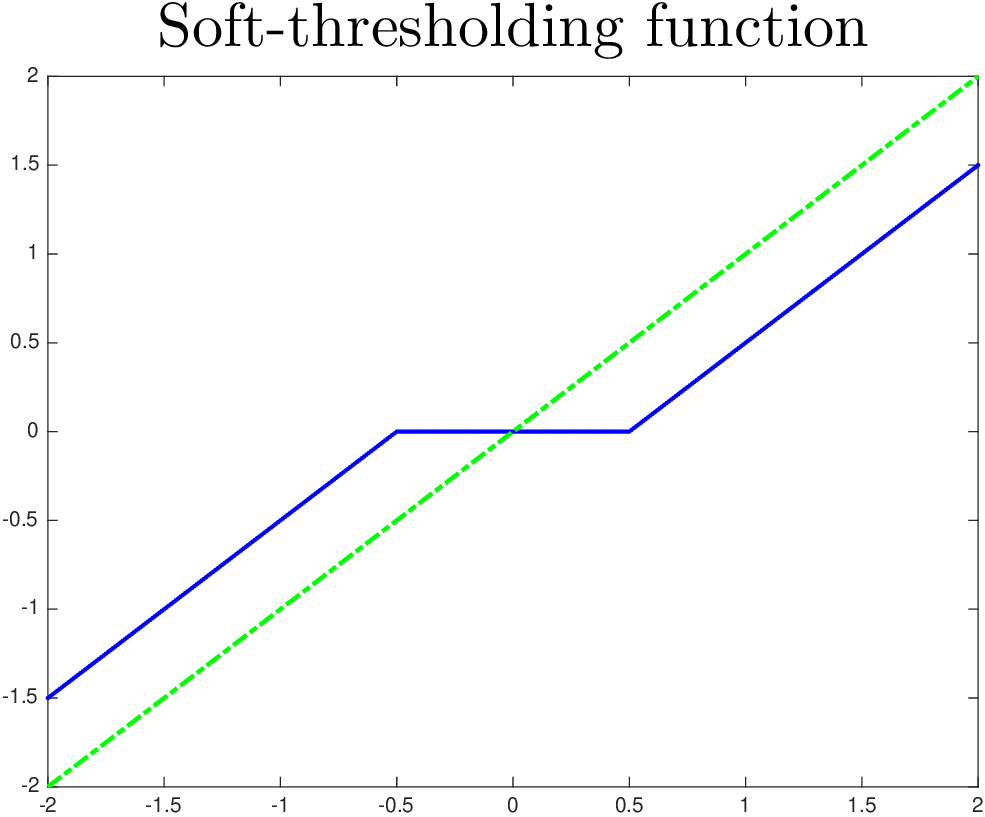}
\end{minipage}  &
\begin{minipage}[t]{0.45\linewidth}
\includegraphics[scale=0.35]{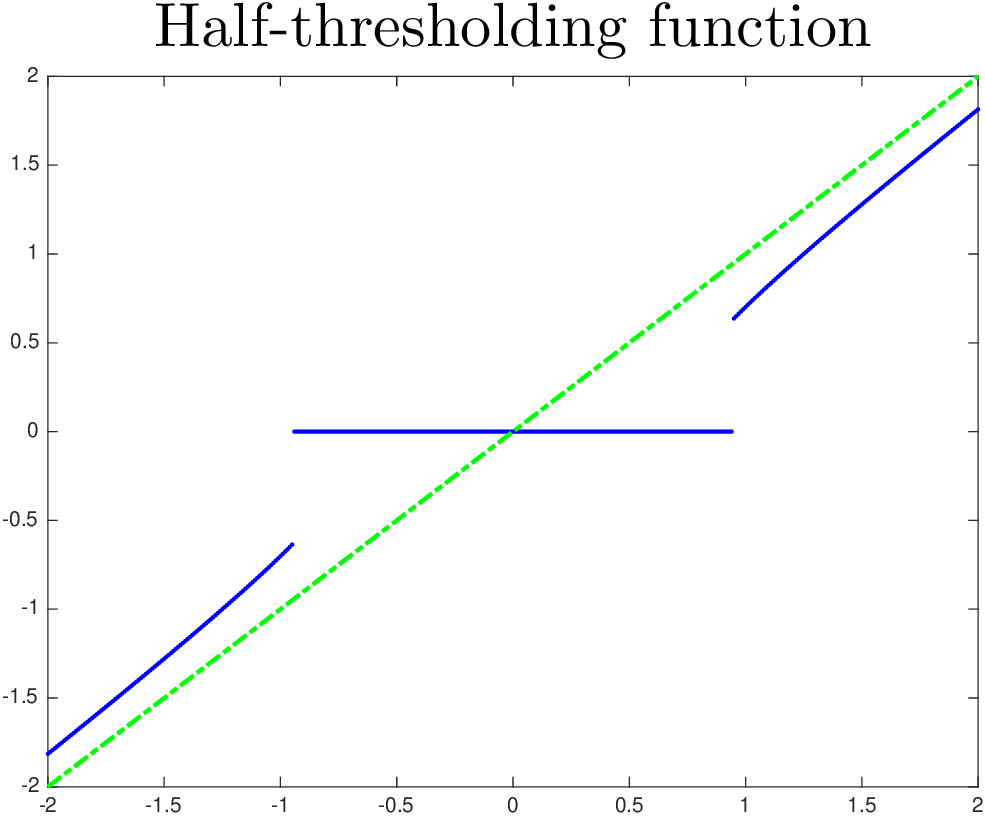}
\end{minipage}  \\
\begin{minipage}[t]{0.45\linewidth}
\includegraphics[scale=0.35]{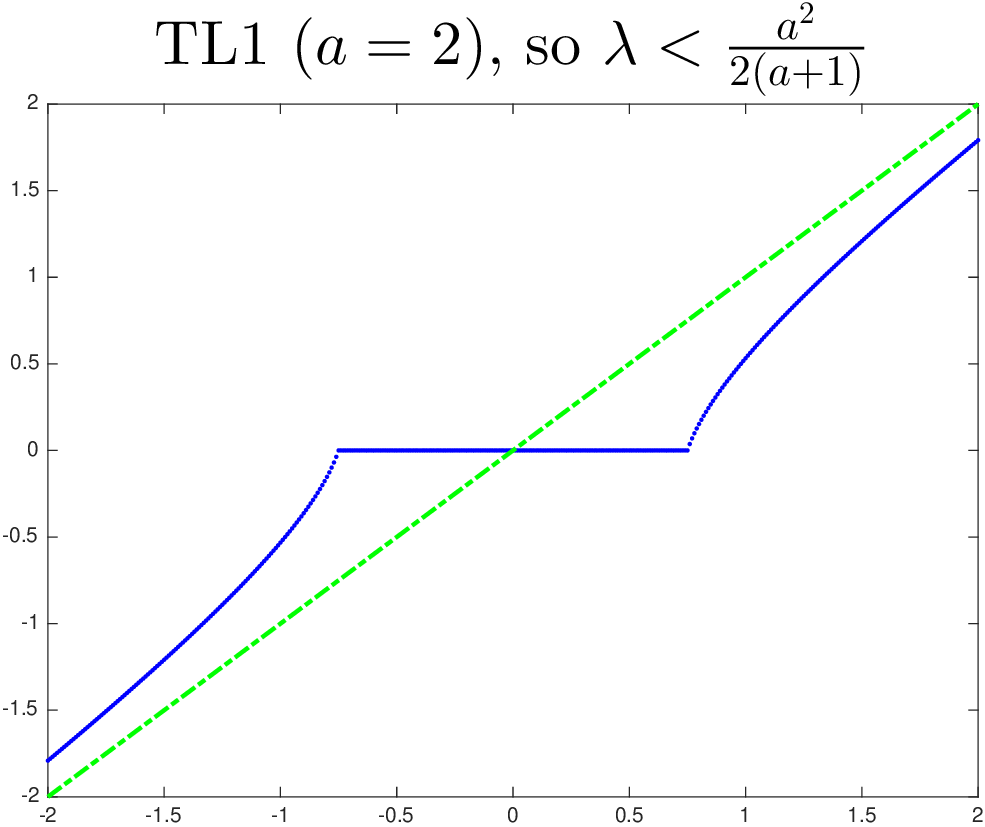}
\end{minipage}  &
\begin{minipage}[t]{0.45\linewidth}
\includegraphics[scale=0.35]{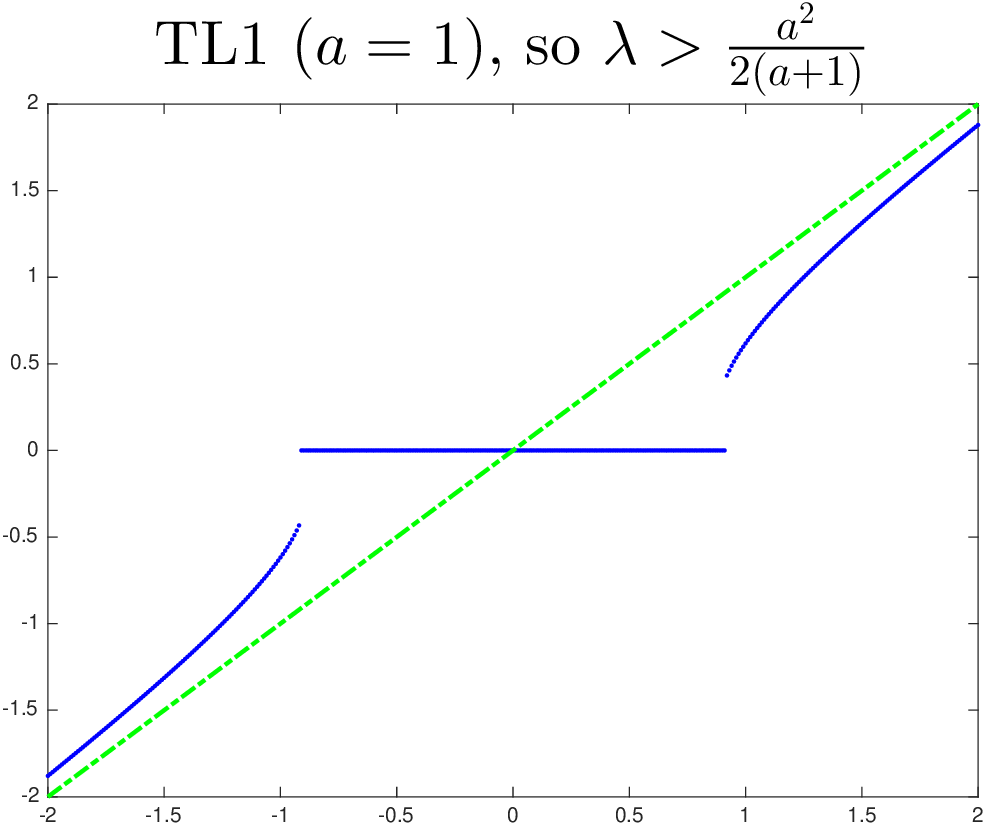}
\end{minipage}  \\
\end{tabular}
\caption{Soft/half (top left/right), TL1 (sub/super critical, lower left/right) 
 thresholding functions at $\lambda = 1/2$.}
\label{figure: threshold plot}
\end{figure}


It is interesting to compare the TL1 thresholding function with the hard/soft thresholding 
function of $l_0$/$l_1$ regularization, and the half thresholding function of 
$l_{1/2}$ regularization. These three  
functions 
(\cite{hard-threshold-blumensath2008iterative,soft-threshold-lp-daubechies2004iterative,xian-half}) are: 
\begin{equation}
H_{\lambda,0}(x) = \left\{ 
\begin{array}{ll}
x, & \ \ |x| > (2\lambda)^{1/2} \\
0, & \ \ {\rm otherwise}
\end{array} \right.
\end{equation} 

\begin{equation}
H_{\lambda,1}(x) = \left\{ 
\begin{array}{ll}
x-sgn(x)\, \lambda, & \ \ |x| > \lambda \\
0, & \ \ {\rm otherwise}
\end{array} \right.
\end{equation} 

and 
\begin{equation}
H_{\lambda,1/2}(x) = \left\{ 
\begin{array}{ll}
f_{2\lambda,1/2}(x), & \ \ |x| > \frac{(54)^{1/3}}{4}(2\lambda)^{2/3} \\
0, & \ \ {\rm otherwise}
\end{array} \right.
\end{equation} 
where $f_{\lambda,1/2}(x) = \frac{2}{3}x \left( 1 + \cos(\frac{2 \pi}{3} 
         - \frac{2}{3} \Phi_{\lambda}(x) ) \right)$ 
and $\Phi_{\lambda}(x) = \arccos(\frac{\lambda}{8}(\frac{|x|}{3})^{-\frac{3}{2}})$.

In \FIG\ref{figure: threshold plot}, we plot the closed-form thresholding formulas 
(\ref{TL1thrfunc}) for $\lambda \leq$ and  $ \lambda > \frac{a^2}{2(a+1)}$ respectively.  
We observe and prove that when  $\lambda < \frac{a^2}{2(a+1)}$, the TL1 
threshold function is continuous (Appendix C), same as soft-thresholding function. 
While if $\lambda > \frac{a^2}{2(a+1)}$, the TL1 thresholding  
function has a jump discontinuity at threshold, similar to half-thresholding function.  
For different threshold scheme, it is believed that continuous formula is more
stable, while discontinuous formula separates nonzero and trivial coefficients
more efficiently and sometimes converges faster  \cite{SparseNet}.

\subsection{Convergence Theory for DFA}
\label{section: convergence}
We establish the convergence theory for direct fixed point iterative algorithm,
similar to \cite{xian-half-theory-2014,xian-half,DCATL1}.
Recall in (\ref{func surrogate}), we introduced two functions $C_{\lambda}(x)$  
(the objective function in TL1 regularization),  and $C_{\mu}(x,z)$. They will appear
in the proof of:
\medskip

\begin{thm}  \label{theorem: convergence}
Let $\{ x^n \}$ be the sequence generated by the iteration scheme 
(\ref{alg: DFA}) under the condition $\|A\|^2 < 1/\mu $. Then: 
\begin{enumerate}
\item[1)] $\{ x^n \}$ is a minimizing sequence of the function $C_{\lambda}(x)$.
If the initial vector $x^0 = 0$ and $\lambda > \frac{\| y\|^2}{2(a+1)}$, 
the sequence $\{ x^n \}$ is bounded. 
\item[2)] $\{ x^n \}$ is asymptotically regular, i.e. 
$\lim \limits_{n \rightarrow \infty} \| x^{n+1} -x^n \| = 0$. 
\item[3)] Any limit point $x^*$ of $\{ x^n \}$ is a stationary point satisfying 
equation (\ref{equa: fixed representation}), that is 
$x^* = G_{\lambda \mu, a}( B_{\mu}(x^*) )$.
\end{enumerate}
\end{thm}
\begin{proof}
\begin{enumerate}
\item[1)]
From the proof of Lemma (\ref{lem 1}), we can see that 
\[ C_{\mu} (x ^{n+1}, x^n) = \min \limits_x C_{\mu} (x,x^n). \]
By the definition of function $C_\lambda(x)$ and $C_{\mu}(x,z)$
 (\ref{func surrogate}), we have the following equation: 
\[ 
	C_{\lambda}(x^{n+1}) = \frac{1}{\mu} \left[  C_{\mu}(x^{n+1}, x^n) 
	- \frac{1}{2} \|x^{n+1} - x^n\|_2^2 \right]
									+  \frac{1}{2} \| Ax^{n+1} - Ax^n\|_2^2
\]
Further since $\|A\|^2 < 1/\mu$, 
\begin{equation} \label{inequ: C_lambda x^n x^n+1}
\begin{array}{ll}
C_{\lambda}(x^{n+1}) 
	& \leq  \frac{1}{\mu} \left\{  C_{\mu}(x^{n}, x^n) 
		- \frac{1}{2} \|x^{n+1} - x^n\|_2^2 \right\}	
		+  \frac{1}{2} \| Ax^{n+1} - Ax^n\|_2^2  \\
	& = C_{\lambda}(x^n) + \frac{1}{2}( \| A(x^{n+1} - x^n)\|_2^2 
		- \frac{1}{\mu} \| x^{n+1} - x^n \|_2^2 ) \\
	& \leq C_{\lambda}(x^n) 
\end{array}
\end{equation}
So we know that sequence $\{ C_{\lambda}(x^n) \}$ is decreasing monotonically. 

In DFA, if we set trivial initial vector $x^0 = 0$ and parameter $\lambda$
satisfying $\lambda > \frac{\| y\|^2}{2(a+1)}$, we show that $\{ x^n\}$
is bounded. Since $\{ C_{\lambda}(x^n) \}$ is decreasing, 
\[ 
	C_{\lambda}(x^n) \leq C_{\lambda}(x^0), \ \ \ \text{for any }  n.  
\]
So we have $ \lambda P_a(x^n) \leq C_{\lambda}(x^0) $. As $\|x^n\|_\infty$
be the largest entry in absolute value of vector $x^n$, 
$ \lambda \rho_a(\|x^n\|_\infty) \leq C_{\lambda}(x^0)$.
Due to the definition of $\rho_a$, it is easy to check that the above inequality
is equivalent to 
\[
\left( \  \lambda(a+1) - C_{\lambda}(x^0) \ \right) \|x^n\|_{\infty} \leq a C_{\lambda} (x^0). 
\]
In order to bound $\{x^n\}$, we need the condition 
$ \lambda > C_{\lambda} (x^0) / (a+1) $.
Especially when $x^0$ is zero, one sufficient condition for $\{ x^n \}$ 
to be bounded is 
\[ 
	\lambda > \frac{\| y\|^2}{2(a+1)}.
\]
\item[2)]
Since $\| A \|^2 < 1/\mu$, we denote $\epsilon = 1 - \mu \| A \|^2 > 0$.
Then we have the inequality 
$\mu \| A(x^{n+1} - x^n)\|^2_2 \leq (1-\epsilon) \| x^{n+1} - x^n \|^2$,
which can be rewritten as 
\[
	 \| x^{n+1} - x^{n}\|^2 \leq \frac{1}{\epsilon} \| x^{n+1} - x^{n}\|^2
	-   \frac{\mu}{\epsilon}  \| A(x^{n+1} - x^n)\|^2_2.
\]
In the above inequality, we sum the index $n$ from $1$ to $N$ and find:
\begin{equation*}
\begin{array}{ll}
\vspace{1mm}
\sum \limits_{n=1}^{N} \| x^{n+1} - x^{n}\|^2 
	& \leq \frac{1}{\epsilon} \sum \limits_{n=1}^{N} \| x^{n+1} - x^{n}\|^2
		- \frac{\mu}{\epsilon}  \sum \limits_{n=1}^{N} \| A(x^{n+1} - x^n)\|^2_2 \\
\vspace{2mm}
	& \leq \frac{\mu}{\epsilon} \sum \limits_{n=1}^{N} 
		2 \left( C_\lambda(x^n) - C_\lambda(x^{n+1} ) \right) \\
	& \leq \frac{2 \mu}{\epsilon} C_\lambda(x^0), 
\end{array}
\end{equation*}
where the last second inequality comes from 
(\ref{inequ: C_lambda x^n x^n+1}) above . 
Thus the infinite sum of sequence $\| x^{n+1} - x^{n}\|^2$ is convergent, which
 implies that 
 \[
 	\lim \limits_{n \rightarrow \infty} \| x^{n+1} - x^{n}\| =  0.
 \]
 
\item[3)]
Denote
$L_{\lambda,\mu}(z,x) = \frac{1}{2} \|z - B_{\mu}(x) \|^2 + \lambda \mu P_a(z)$
and 
\[
	D_{\lambda,\mu}(x) = L_{\lambda,\mu}(x,x) - \min \limits_z L_{\lambda,\mu}(z,x).
\]
By its definition and the proof of Lemma \ref{lem 1} 
(especially (\ref{equa: optimal proximal})), we have $D_{\lambda,\mu}(x) \geq 0$ 
and
\[
	D_{\lambda,\mu}(x) = 0 \text{ if and only if $x$ satisfies (\ref{equa: fixed representation})}. 
\]
Assume that $x^*$ is a limit point of $\{x^n\}$ and a subsequence of $x^n$ (still denoted the same)
converges to it. Because of DFA iterative scheme (\ref{alg: DFA}), we have 
$x^{n+1} = arg \min_z L_{\lambda,\mu}(z,x^n)$, which implies that
\begin{equation*}
\begin{array}{l}
D_{\lambda,\mu}(x^n) 
\vspace{1mm}
	 = L_{\lambda,\mu}(x^n,x^n) - L_{\lambda,\mu}(x^{n+1},x^n) \\
	 = \lambda \mu \left( P_a(x^n) - P_a(x^{n+1}) \right) 
		- \frac{1}{2} \| x^{n+1} - x^n \|^2
		+ \left\langle \mu A^t (Ax^n - y), x^n - x^{n+1} \right\rangle 
\end{array}
\end{equation*}
Thus we know
\begin{equation*}
\begin{array}{l}
\vspace{1mm}
	\lambda P_a(x^n) - \lambda P_a(x^{n+1}) \\
	= \frac{1}{2\mu} \|x^{n+1}- x^n\|^2
		+ \frac{1}{\mu} D_{\lambda,\mu}(x^n) 
		+ \left\langle A^t (Ax^n - y), x^n - x^{n+1} \right\rangle, 
\end{array}
\end{equation*}
from which we get
\begin{equation*}
\begin{array}{ll}
C_{\lambda}(x^n) - C_{\lambda}(x^{n+1}) 
\vspace{1mm}
	& = \lambda P_a(x^n) - \lambda P_a(x^{n+1}) 
		+ \frac{1}{2} \|Ax^n-y\|^2 - \frac{1}{2} \|Ax^{n+1}-y\|^2 \\
\vspace{1mm}
	& =  \frac{1}{2\mu} \|x^{n+1}- x^n\|^2 + \frac{1}{\mu} D_{\lambda,\mu}(x^n)
		- \frac{1}{2} \| A(x^n - x^{n+1}) \|^2_2 \\
	& \geq  \frac{1}{\mu} D_{\lambda,\mu}(x^n) 
		+ \frac{1}{2}(\frac{1}{\mu} - \|A\|^2) \| x^n - x^{n+1}\|^2.
\end{array}
\end{equation*}
So $0 \leq D_{\lambda,\mu}(x^n)  \leq 
	\mu \left( C_{\lambda}(x^n) - C_{\lambda}(x^{n+1})  \right)$.
Also we know from part (1) of this theorem that $\{ C_{\lambda}(x^n)\}$ 
converges, so $\lim \limits_{n \rightarrow \infty} D_{\lambda,\mu}(x^n) = 0$. 
Thus as the limit point of the sequence $x^n$, the point 
$x^*$ satisfies equation (\ref{equa: fixed representation}).

\end{enumerate}
\end{proof}

\subsection{Semi-Adaptive Thresholding Algorithm --- TL1IT-s1} 
\label{section: scheme 1}
In the following 2 subsections, we present two adaptive parameter TL1 algorithms. 
We begin with formulating an optimality condition on the 
regularization parameter $\lambda$, 
which serves as the basis for 
parameter selection and updating in the semi-adaptive algorithm. 

Let us consider the so called $k$-sparsity problem for  
(\ref{model: uncons-optim}). The solution is $k$-sparse by prior knowledge or estimation. 
For any $\mu$, denote $B_{\mu}(x) = x + \mu A^T(b-Ax)$ and $|B_{\mu}(x)|$ 
is the vector from taking absolute value 
of each entry of $B_{\mu}(x)$. Suppose that $x^*$ is the TL1 solution, and 
without loss of generality,  
$|B_{\mu}(x^*)|_1 \geq |B_{\mu}(x^*)|_2 \geq ... \geq |B_{\mu}(x^*)|_N$. 
Then, the following inequalities hold: 
\begin{equation}
\begin{array}{lll}
|B_{\mu}(x^*)|_i > t & \Leftrightarrow &  i \in \{ 1,2,...,k \}, \\
|B_{\mu}(x^*)|_j \leq t & \Leftrightarrow &  j \in \{ k+1,k+2,...,N \}, 
\end{array}
\end{equation}
where $t$ is our threshold value.

Recall that $t^*_3 \leq t \leq t^*_2$. So 
\begin{equation}
\begin{array}{l}
|B_{\mu}(x^*)|_{k} \geq t \geq t^*_3 = \sqrt{2\lambda \mu (a+1)} - \frac{a}{2}; \\
|B_{\mu}(x^*)|_{k+1} \leq t \leq t^*_2 = \lambda \mu \frac{a+1}{a}. 
\end{array}
\end{equation}
It follows that
\begin{equation*}
\lambda_1 \equiv \dfrac{a |B_{\mu}(x^*)|_{k+1}}{\mu (a+1)} \leq \lambda 
\leq \lambda_2 \equiv \dfrac{(a+2|B_{\mu}(x^*)|_{k})^2}{8(a+1)\mu} 
\end{equation*}
or $\lambda^* \in [\lambda_1, \lambda_2]$.

\begin{algorithm}[h]
\caption{TL1 Thresholding Algorithm ---  TL1IT-s1}
 \textbf{Initialize:} 
 $x^0$; \ \ $\mu_0= \frac{(1-\varepsilon)}{\|A\|^2}$ and $a$\;
 \While{not converged}{
  $\mu  = \mu_0$; \ \	$ z^n := B_{\mu}(x^n) = x^n + \mu A^T(y-Ax^n) $; \\
  $\lambda_{1}^{n} = \dfrac{a |z^n|_{k+1}}{\mu (a+1)}$; \ \
  $\lambda_{2}^{n} = \dfrac{(a+2|z^n|_{k})^2}{8(a+1)\mu }$\;
  \eIf{$\lambda_{1}^{n} \leq \frac{a^2}{2(a+1)\mu}$}{
   $\lambda = \lambda_{1}^{n}$; \ \ 
   $t = \lambda \mu \frac{a+1}{a}$\;
   for i = 1:length(x)\:  \\
   \ \ \  if $|z^n(i)| > t$, then
   $x^{n+1}(i) = g_{\lambda \mu }(z^n(i))$; \\
   \ \ \ if $|z^n(i)| \leq t$, then
   $x^{n+1}(i) = 0$.
   }{
   $\lambda = \lambda_{2}^{n}$; \ \ 
   $t = \sqrt{2\lambda \mu (a+1)} - \frac{a}{2}$ \;
   for i = 1:length(x)\: \\
   \ \ \ if $|z^n(i)| > t$, then
   $x^{n+1}(i) = g_{\lambda \mu}(z^n(i))$; \\
   \ \ \ if $|z^n(i)| \leq t$, then
   $x^{n+1}(i) = 0$.
  }
  $n \rightarrow n+1$\;
 }
\end{algorithm}

The above estimate helps to set optimal regularization 
parameter. A choice of $\lambda^*$ is
\begin{equation} \label{equ: lambda}
\lambda^* = \left\{
\begin{array}{ll}
\lambda_1, & \quad \text{if} \ \ 
   \lambda_1 \leq \frac{a^2}{2(a+1)\mu}, \ \ \text{then} \ \ \lambda^* \leq \frac{a^2}{2(a+1)\mu} 
                                         \Rightarrow t = t^*_2; \\
\lambda_2, & \quad \text{if} \ \ 
   \lambda_1 > \frac{a^2}{2(a+1)\mu},    \ \ \text{then} \ \ \lambda^* > \frac{a^2}{2(a+1)\mu}
                                        \Rightarrow t = t^*_3.
\end{array} \right.
\end{equation}
In practice, we approximate $x^*$ by $x^n$ in 
(\ref{equ: lambda}), so 
\[ 
	\lambda_1 = \dfrac{a |B_{\mu}(x^n)|_{k+1}}{\mu (a+1)},  \ \ \ \   
	\lambda_2 = \dfrac{(a+2|B_{\mu}(x^n)|_{k})^2}{8(a+1)\mu},
\]
at each iteration step. So we have an adaptive iterative algorithm 
without pre-setting the regularization parameter $\lambda$. 
Also the TL1 parameter $a$ is still free (to be selected), 
thus this algorithm is overall semi-adaptive, which is named 
TL1IT-s1 for short and summarized in Algorithm 1.

\subsection{Adaptive Thresholding Algorithm --- TL1IT-s2}
For TL1IT-s1 algorithm, at each iteration step, it is required to compare 
$\lambda_n$ and $\frac{a^2}{2(a+1)\mu}$. Here instead, we vary TL1 
parameter `a' and choose $a=a_n$ in each iteration, 
such that the inequality  $\lambda_n \leq \frac{a^2_n}{2(a_n +1)\mu_n}$ holds.

The thresholding scheme is now simplified to just one 
threshold parameter $t = t^*_2$. 
Putting $\lambda = \frac{a^2}{2(a +1) \mu}$ at critical value, 
the parameter $a$ is expressed as: 
\begin{equation} \label{form: scheme 2 a lambda}
a = \lambda\mu + \sqrt{(\lambda \mu)^2 + 2 \lambda \mu }.
\end{equation}
The threshold value is: 
\begin{equation} \label{form: scheme 2 t}
t  =  t^*_2 = \lambda \mu \frac{a+1}{a} 
   = \frac{\lambda\mu}{2} + \frac{\sqrt{(\lambda \mu)^2 + 2 \lambda \mu }}{2}.
\end{equation}

Let $x^*$ be the TL1 optimal solution. Then we have the following inequalities: 
\begin{equation}
\begin{array}{lll}
|B_{\mu}(x^*)|_i > t & \Leftrightarrow &  i \in \{ 1,2,...,k \}, \\
|B_{\mu}(x^*)|_j \leq t & \Leftrightarrow &  j \in \{ k+1,k+2,...,N \}.
\end{array}
\end{equation}
So, for parameter $\lambda$, we have: 
$$ \dfrac{1}{\mu} \dfrac{2|B_{\mu}(x^*)|_{k+1}^2}{1+2|B_{\mu}(x^*)|_{k+1}} \leq \lambda 
\leq \dfrac{1}{\mu} \dfrac{2|B_{\mu}(x^*)|_{k}^2}{1+2|B_{\mu}(x^*)|_{k}}.$$
Once the value of $\lambda$ is determined, the parameter $a$ is 
given by (\ref{form: scheme 2 a lambda}).

In the iterative method, we approximate the optimal solution $x^*$ by 
$x^n$. The resulting parameter selection is:
\begin{equation}
\begin{array}{l}
\lambda_n = \dfrac{1}{\mu_n} \dfrac{2|B_{\mu_n}(x^*)|_{k+1}^2}{1+2|B_{\mu_n}(x^*)|_{k+1}}; \\
a_n = \lambda_n \mu_n + \sqrt{(\lambda_n \mu_n)^2 + 2 \lambda_n \mu_n }.
\end{array}
\end{equation}

In this algorithm (TL1IT-s2 for short), only parameter $\mu$ is fixed and 
$\mu \in (0, \|A\|^{-2})$. The summary is below (Algorithm 2).  

\begin{algorithm}[h]
\caption{Adaptive TL1 Thresholding Algorithm ---  TL1IT-s2}
 \textbf{Initialize:} 
 $x^0$, $\mu_0= \frac{(1-\varepsilon)}{\|A\|^2}$\;
 \While{not converged}{
  $\mu = \mu_0$; \ \ $ z^n := x^n + \mu\, A^T(y-Ax^n) $\;
  $\lambda_n = \dfrac{1}{\mu}\dfrac{2|z^n_{k+1}|^2}{1+2|z^n_{k+1}|}$\;
  \vspace{1mm}
  $a_n = \lambda_n \mu + \sqrt{(\lambda_n \mu )^2 + 2 \lambda_n \mu }$\;  
  \vspace{1mm}
  $t = \frac{\lambda_n \mu}{2} + \frac{\sqrt{(\lambda_n \mu)^2 + 2 \lambda_n \mu }}{2}$\;
  \vspace{1mm}
  for i = 1:length(x) \\ 
   \ \ \ if $|z^n(i)| > t$, then
   $x^{n+1}(i) = g_{\lambda_n \mu }(z^n(i))$; \\
   \ \ \ if $|z^n(i)| \leq t$, then
   $x^{n+1}(i) = 0$. \\
  $n \rightarrow n+1$\;
 }
\end{algorithm}



\section{Numerical Experiments} \label{section:experiment}
\setcounter{equation}{0}
In this section, we carried out a series of numerical experiments to demonstrate the  
performance of the TL1 thresholding algorithm: semi-adaptive TL1IT-s1. 
All the experiments here are conducted by applying our algorithm to sparse 
signal recovery in compressed sensing. 
Two classes of randomly generated sensing matrices are used 
to compare our algorithms with the  
state-of-the-art iterative non-convex thresholding solvers:   
\textbf{Hard-thresholding} \cite{hard-sparsify-blumensath2012accelerated}, 
\textbf{Half-thresholding} \cite{xian-half}. Here all these thresholding algorithms
need a sparsity estimation to accelerate convergence. Also the Hard 
Thresholding algorithm (AIHT) in \cite{hard-sparsify-blumensath2012accelerated}  
has an additional double over-relaxation step for significant 
speedup in convergence. In the following run time comparison of the three algorithms,
AIHT is clearly the most efficient under the uncorrelated Gaussian sensing matrix. 

We also tested on the adaptive scheme: TL1IT-s2. However, its performance 
is always no better than TL1IT-s1, and so its results are not shown here.
We suggest to use TL1IT-s1 first in CS applications.  
That TL1IT-s2 is not as competitive as TL1IT-s1 may be attributed to its 
limited thresholding scheme. Utilizing double thresholding schemes is 
helpful for TL1IT.  We noticed in our computations that at the beginning
of iterations, the $\lambda_n$'s cross the critical value 
$\frac{a^2}{2(a+1)\mu}$ frequently. Later on, they tend to stay on one side, depending on 
the sensing matrix $A$. However, the sub-critical threshold 
is used for all $A$'s in TL1IT-s2. 

Here we compare only the non-convex iterative thresholding methods, 
and did not include the soft-thresholding algorithm.
 The two classes of random matrices are: 
\begin{itemize}
   \item[1)] Gaussian matrices.
   \item[2)] Over-sampled discrete cosine transform (DCT) matrices with factor $F$.
\end{itemize}

All our tests were performed on a $Lenovo$ desktop: 16 GB of RAM and Intel Core processor  
$i7-4770$ with CPU at $3.40GHz \times 8 $ under 64-bit Ubuntu system.

The TL1 thresholding algorithms do not guarantee a global minimum in general, 
due to nonconvexity. Indeed we observed that TL1 thresholding with 
random starts may get stuck at local minima especially when the matrix $A$ 
is ill-conditioned (e.g. $A$ has a large condition number or is highly coherent). 
A good initial vector $x^0$ is important for thresholding algorithms. 
In our numerical experiments, instead of having $x^0 = 0$ or random, 
we apply YALL1 (an alternating direction $l_1$ method, \cite{yall1}) 
a number of times, e.g.  20 times, to produce a better initial guess $x^0$. 
This procedure is similar to algorithm DCATL1 \cite{DCATL1} initiated at 
zero vector so that the first step of DCATL1 
reduces to solving an unconstrained $l_1$ regularized problem.
For all these iterative algorithms, we implement a unified stopping criterion as
$ \frac{\| x^{n+1} - x^{n} \|}{\|x^n\|} \leq 10^{-8} $ or maximum iteration step equal to 3000.    

\begin{figure}[h]
\centering
\includegraphics[scale=0.45]{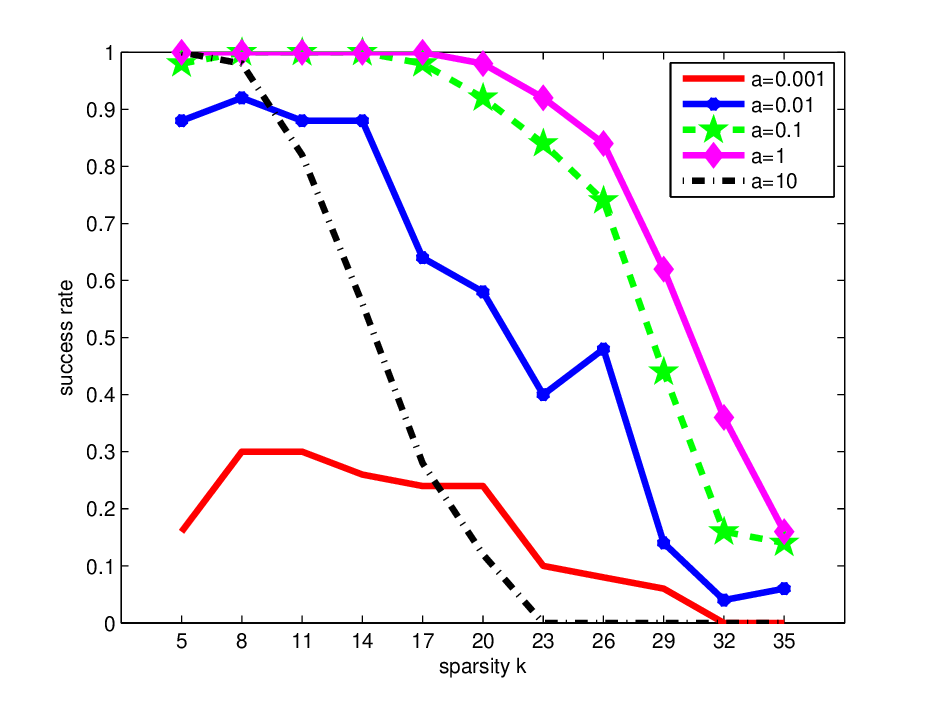}
\caption{Sparse recovery success rates for selection of parameter $a$ with $128\times 512$ Gaussian 
random matrices and TL1IT-s1 method.}
\label{fig:choice of a}
\end{figure}

\subsection{Optimal Parameter Testing for TL1IT-s1}
In TL1IT-s1, the parameter `$a$' is still free. 
When `$a$' tends to zero, the penalty function approaches the $l_0$ norm. 
We tested TL1IT-s1 on sparse vector recovery with different `$a$' values, varying among 
\{0.001, 0.01, 0.1, 1, 100 \}. 
In this test, matrix $A$ is a $128 \times 512$ random matrix, 
generated by multivariate normal distribution $\sim \mathcal{N}(0,\Sigma)$. 
Here the covariance matrix
$\Sigma = \{ 1_{(i=j)} + 0.2\times  1_{(i \neq j)} \}_{i,j}.$
The true sparse vector $x^*$ is also randomly generated under Gaussian distribution,
with sparsity $k$ from the set $\{ 8,  \ 10 , \ 12 , \ \cdots , \ 32 \}$.

For each value of `$a$', we conducted 100 test runs with 
different samples of $A$ and ground truth vector $x^*$. The recovery 
is successful if the relative error: 
$ \frac{ \|x_r - x^*\|_2 }{\|x^*\|_2} \leq 10^{-2} $.

Figure (\ref{fig:choice of a}) shows the success rate vs. sparsity using TL1IT-s1
over 100 independent trials for various parameter $a$ and sparsity $k$. 
We see that the algorithm with $a=1$ is the best among all tested parameter values. 
Thus in the subsequent computation, we set the parameter $a=1$. 
The parameter $\mu = \frac{0.99}{\|A\|^2}$.

\subsection{Signal Recovery without Noise}
\subsubsection*{Gaussian Sensing Matrix}
The sensing matrix $A$ is drawn from $\mathcal{N}(0,\Sigma)$, 
the multi-variable normal distribution with covariance matrix  
$\Sigma = \{ (1-r)\, 1_{(i=j)} + r  \}_{i,j} $,
where $r$ ranges from 0 to 0.8. The larger parameter $r$ is, the more difficult it is 
to recover the sparse ground truth vector. The matrix $A$ is 
$128\times 512$, and the sparsity $k$ varies among $\{ 5 , \ 8 , \ 11 , \cdots , \ 35 \}$. 

\begin{figure}[h]
\begin{tabular}{ll}
\begin{minipage}[t]{0.4\linewidth}
\centering 
\includegraphics[scale=0.325]{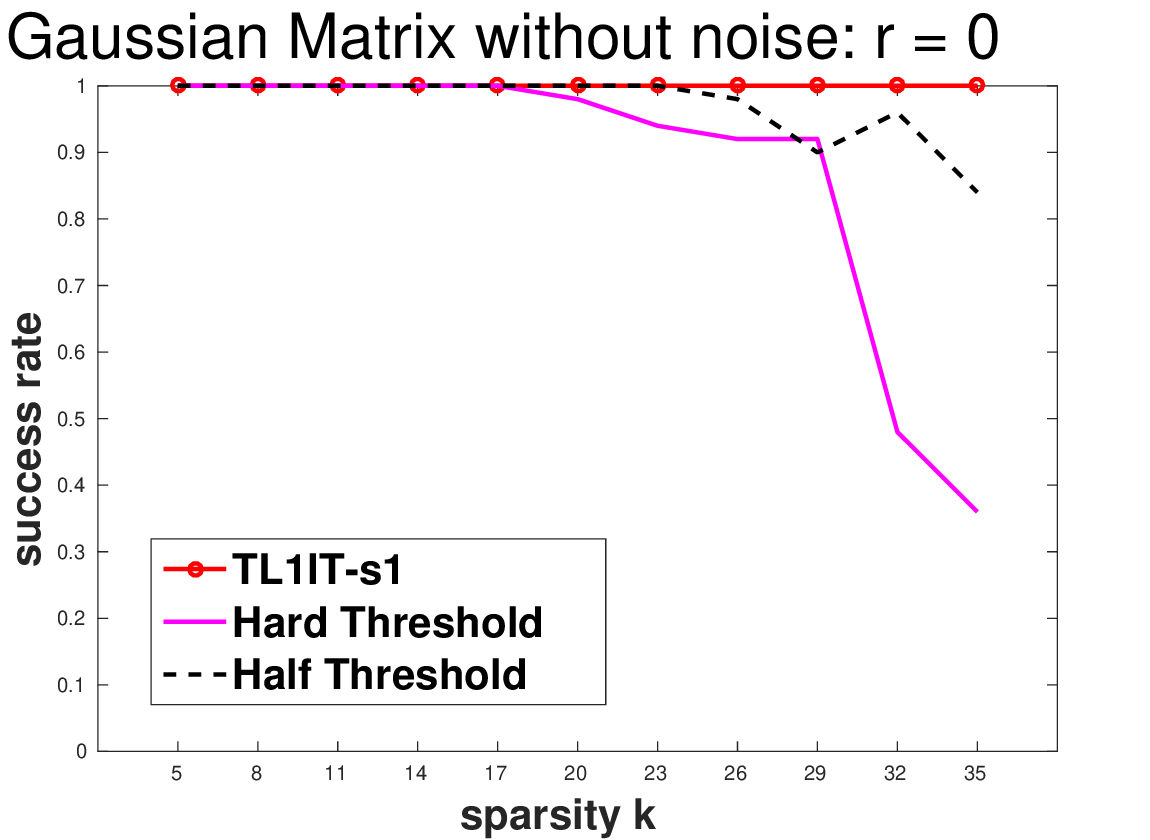}
\end{minipage}  &
\begin{minipage}[t]{0.4\linewidth}
\centering 
\includegraphics[scale=0.325]{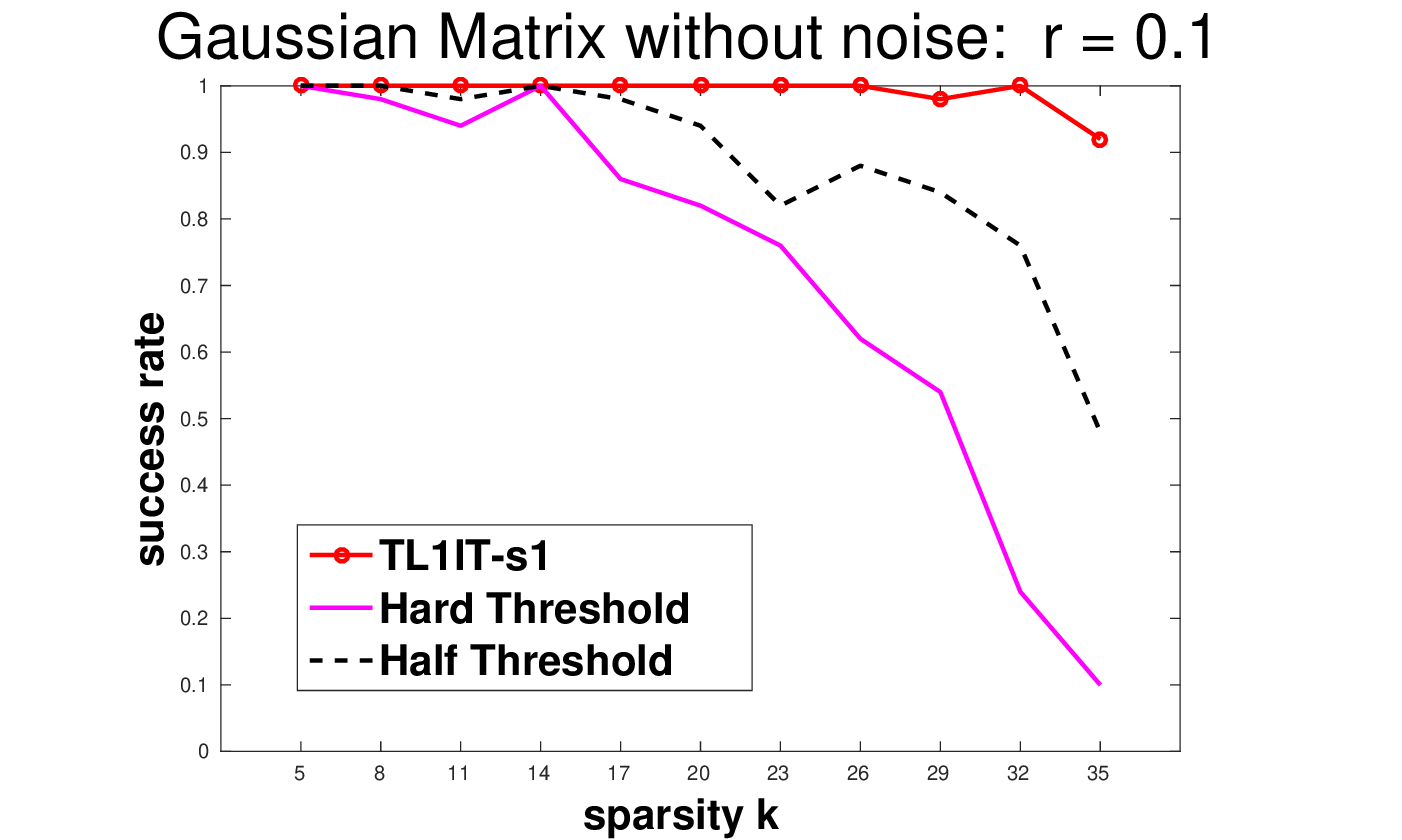}
\end{minipage}  \\
\begin{minipage}[t]{0.4\linewidth}
\centering
\includegraphics[scale=0.325]{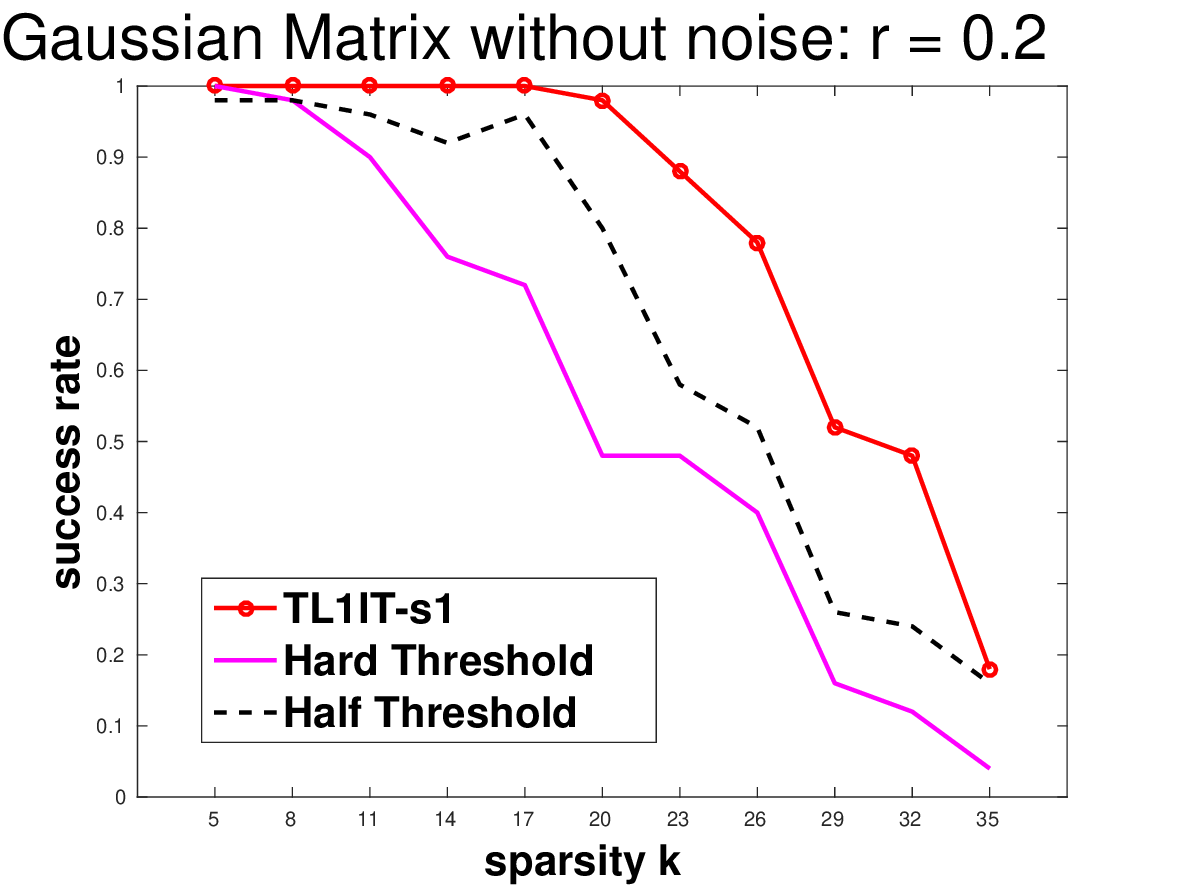}
\end{minipage}  &
\begin{minipage}[t]{0.4\linewidth}
\centering 
\includegraphics[scale=0.325]{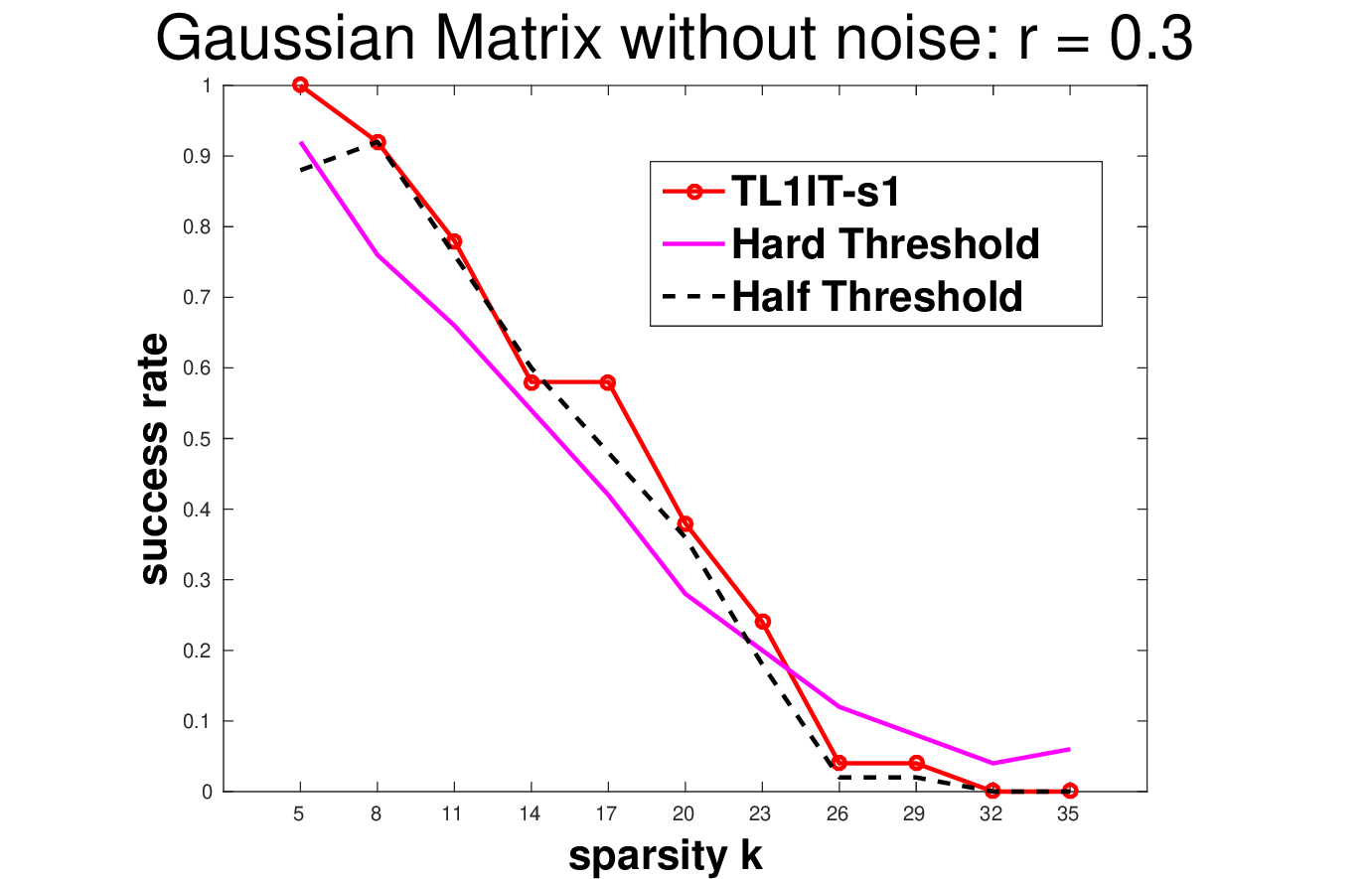}
\end{minipage}  
\end{tabular}
\caption{Sparse recovery algorithm comparison for $128\times 512$ 
Gaussian sensing matrices without measurement noise at covariance parameter $r=0,\, 0.1, \, 0.2, \, 0.3$.}
\label{figure:Gaussian}
\end{figure}  

We compare the three IT algorithms in terms of success rate averaged over 50 random trials. 
A success is recorded if the relative error of recovery is less than $0.001$. 
The success rate of each algorithm is plotted in 
Figure \ref{figure:Gaussian} with parameter $r$ from  
the set: $\{ 0 , \ 0.1 , \ 0.2 , \ 0.3 \}$. 

We see that all three algorithms can accurately recover the signal 
when $r$ and sparsity $k$ are both small. However, the success rates decline, along with 
the increase of $r$ and sparsity $k$. 
At $r= 0$, the TL1IT-s1 scheme recovers almost all testing signals from different sparsity. 
Half thresholding algorithm maintains nearly the same high success rates with a slight decrease 
when $k \geq 26$. 
At $r = 0.3$, TL1IT-s1 leads the half thresholding algorithm with a small margin. 
In all cases, TL1IT-s1 outperforms the other two, while the half thresholding algorithm 
is the second.   

\subsubsection*{Comparison of time efficiency under Gaussian measurements}
One interesting question is about the time efficiency for different thresholding
algorithms. As seen from Figure \ref{figure:Gaussian}, almost all the 3
algorithms, under Gaussian matrices with covariance parameter 
$r = 0$ and sparsity $k = 5,\cdots, 20$, 
achieve 100 \% success recovery. 
So we measured the average convergent time over 20 random tests in the
above situation (see Table $1$), where all the parameters are 
tuned to obtain relative errors around $10^{-5}$.

\begin{table}  \label{table: Time}
\centering 
\begin{tabular}{l cccccc} 
\hline\hline 
sparsity & 5 & 8 & 11 & 14 & 17 & 20 \\ [0.5ex] 
\hline 
TL1IT-s1 & 0.031 & 0.054 & 0.047 & 0.055 & 0.053 & 0.059 \\
Hard 		 & \textbf{0.003} & \textbf{0.003} & \textbf{0.005} 
& \textbf{0.006} & \textbf{0.007} & \textbf{0.007} \\
Half 	     & 0.019	 & 0.017 & 0.017 & 0.023 & 0.020 & 0.025 \\
[1ex] 
\hline 
\end{tabular} 
\caption{Time efficiency (in sec) comparison for 3 algorithms 
under Gaussian matrices.} 
\end{table} 

From the table, we know that Hard Thresholding algorithm costs the least time
among all three. So under this uncorrelated normal distribution measurement, 
Hard Thresholding algorithm is the most efficient, with Half Thresholding algorithm
the second. Though TL1IT-s1 has the lowest relative error in recovery, it takes more time. 
One reason is that TL1IT-s1 iterations go between two thresholding schemes, 
which makes it more adaptive to data for a higher computational cost. 

\subsubsection*{Over-sampled DCT Sensing Matrix} 
The over-sampled DCT matrices \cite{Fann_12,l1-l2-lou2014computing} are: 
\begin{equation}
\begin{array}[c]{l}
   A  =  [a_1,...,a_N] \in \Re^{M \times N} \\
   {\rm where} \ \ a_j = \dfrac{1}{\sqrt{M}} cos(\dfrac{2 \pi \omega (j-1)}{F}), \ \ j = 1,...,N,   \\
   \text{ and $\omega$ is a random vector, drawn uniformly from $(0,1)^M$.}     
\end{array}
\end{equation}
Such matrices appear as the real part of the complex discrete Fourier matrices in 
spectral estimation and super-resolution problems \cite{Super:candes2013mini,Fann_12}. 
An important property is their high coherence measured by the maximum of 
absolute value of cosine of the angles between each pair of column vectors of $A$. 
For a $100 \times 1000$ over-sampled DCT matrix at $F =10$, the coherence is about 0.9981, while 
at $F=20$ the coherence of the same size matrix is typically 0.9999.

\begin{figure}
\begin{tabular}{lll}
\begin{minipage}[t]{0.4\linewidth}
\includegraphics[scale=0.325]{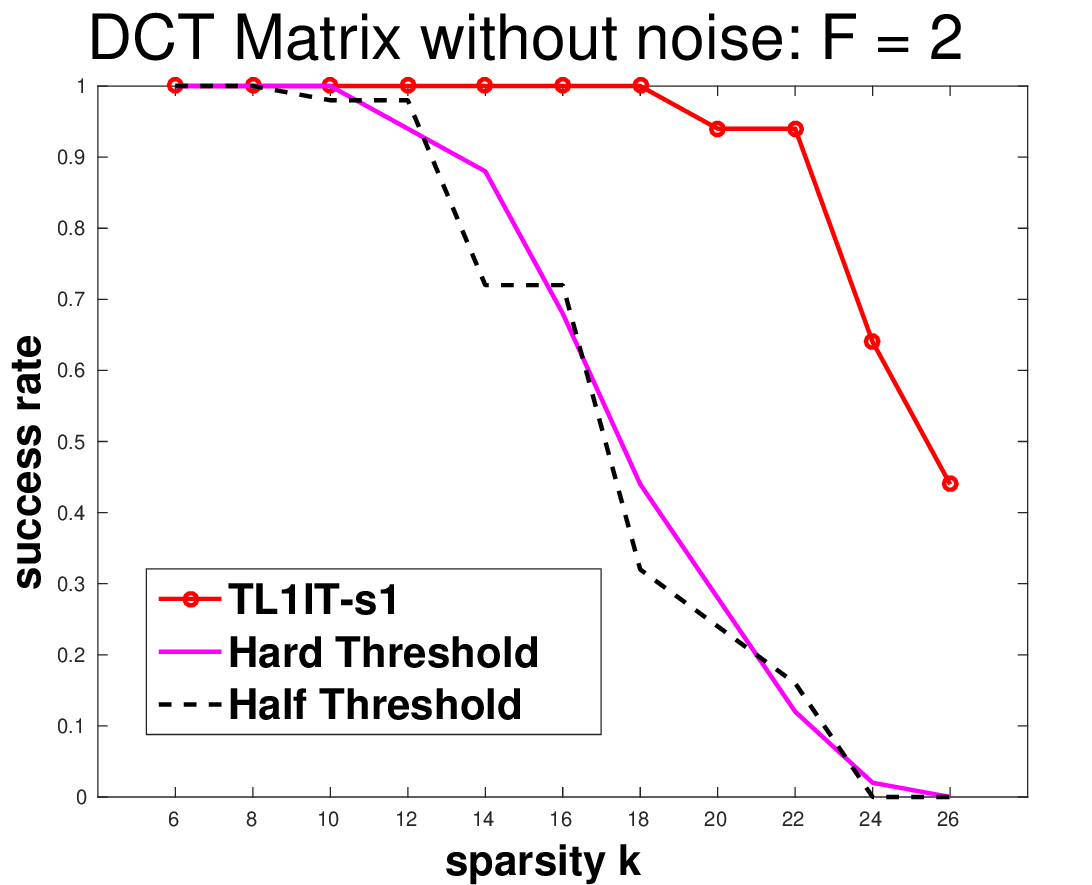}
\end{minipage}  &
\begin{minipage}[t]{0.4\linewidth}
\includegraphics[scale=0.325]{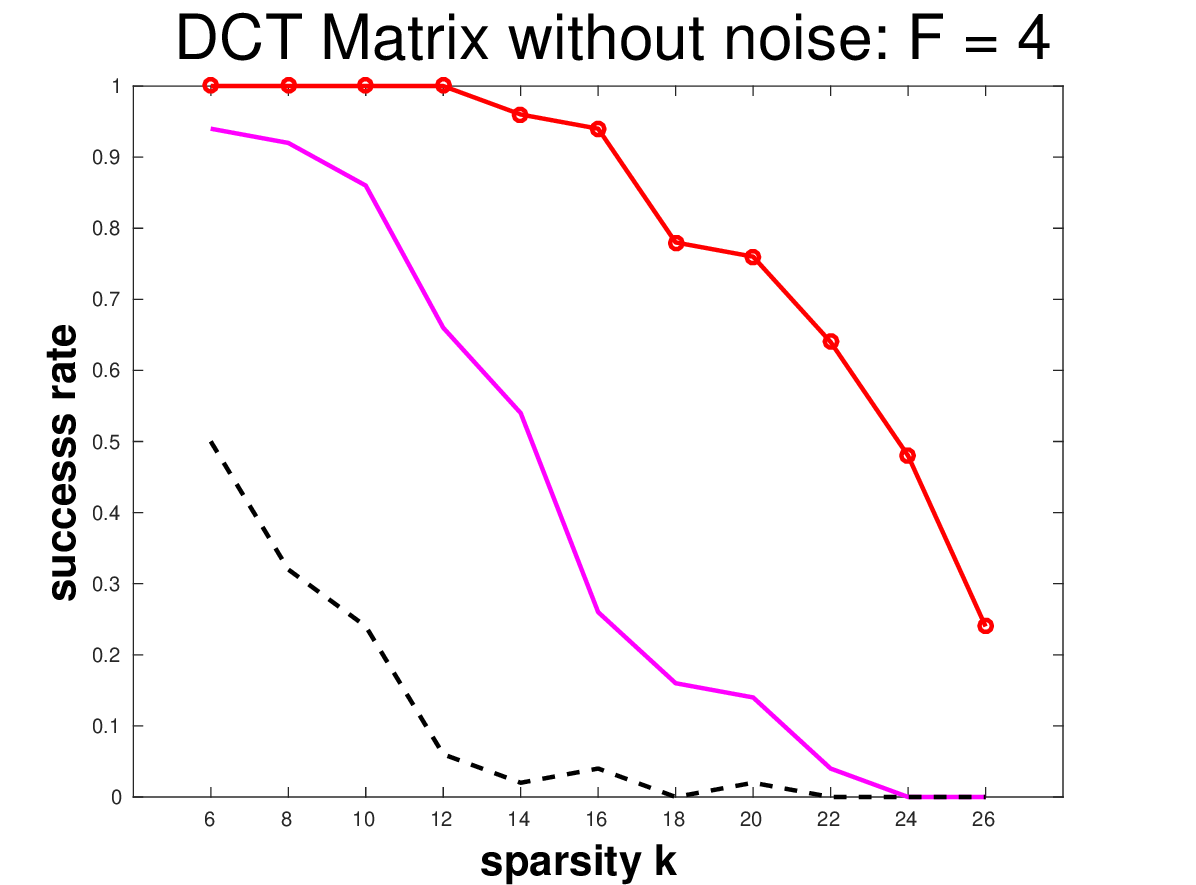}
\end{minipage}  \\
\begin{minipage}[t]{0.4\linewidth}
\includegraphics[scale=0.325]{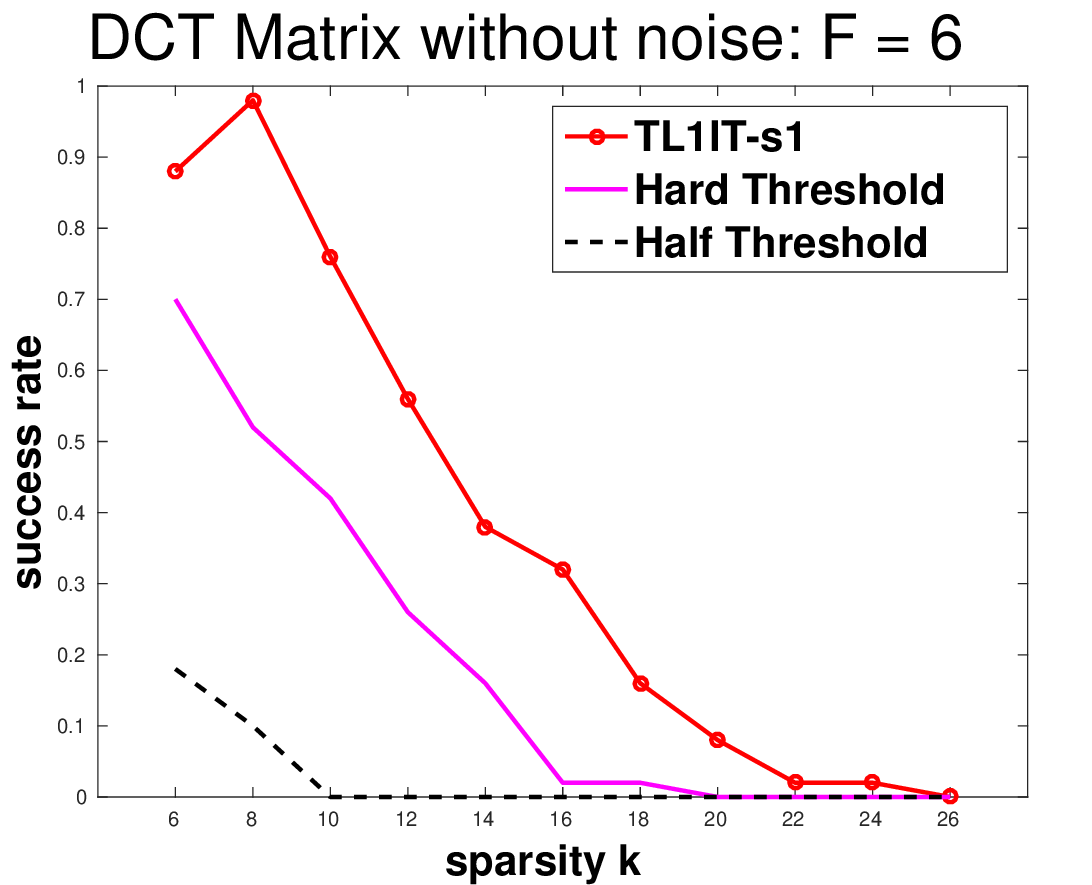}
\end{minipage}  &
\begin{minipage}[t]{0.4\linewidth}
\includegraphics[scale=0.325]{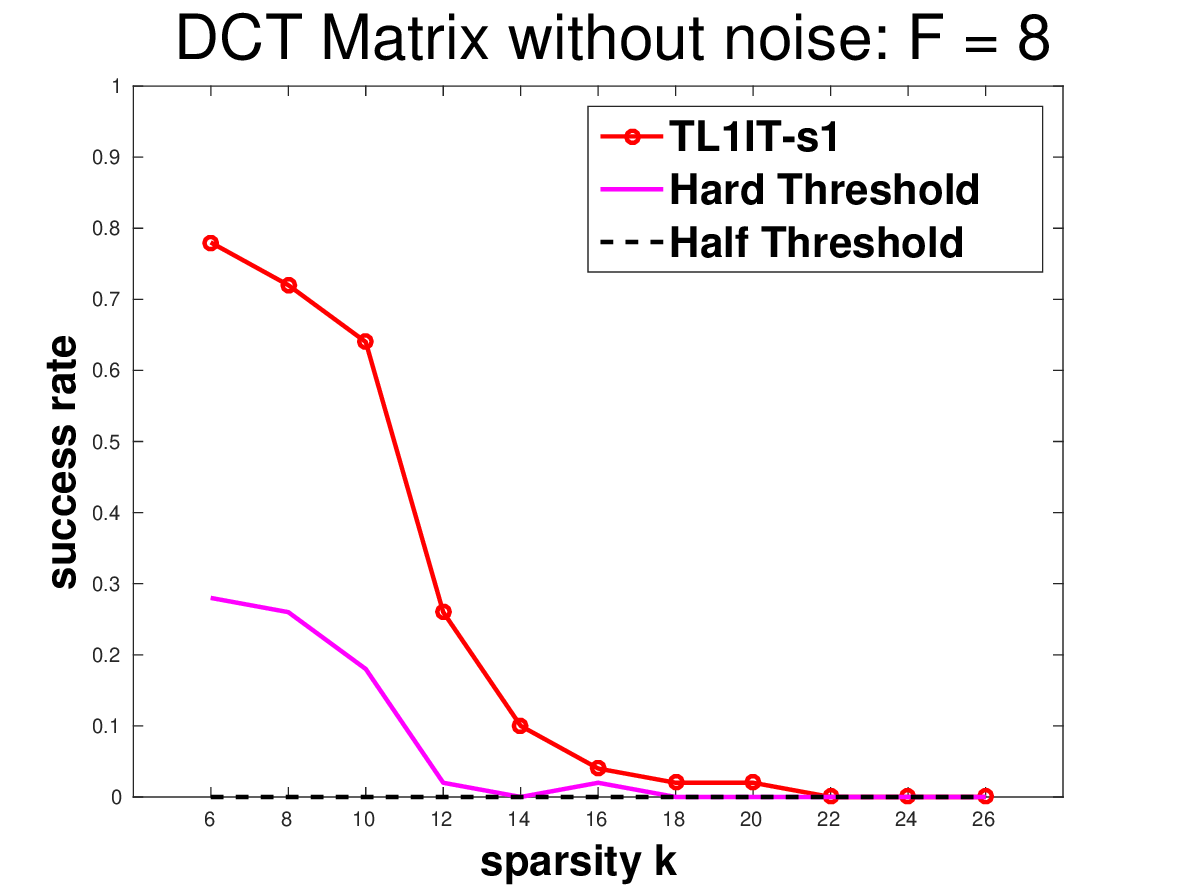}
\end{minipage} 
\end{tabular}
\caption{Algorithm comparison for $100\times 1500$ over-sampled DCT random matrices without noise at different factor $F$.}
\label{figure:DCT}
\end{figure}   

The sparse recovery under such matrices is possible only if the non-zero elements 
of solution $x$ are sufficiently separated. This phenomenon is characterized as 
$minimum$ $separation$ in \cite{Super:candes2013mini}, with minimum length  
referred as the Rayleigh length (RL). The value of RL for matrix $A$ is equal to 
the factor $F$. It is closely related to the coherence in the sense that larger $F$ 
corresponds to larger coherence of a matrix. We find empirically that at least 2RL is 
necessary to ensure optimal sparse recovery with spikes further apart for more 
coherent matrices.

Under the assumption of sparse signal with $2RL$ separated spikes, we compare 
the four non-convex IT algorithms in terms of success rate. 
The sensing matrix $A$ is of size $100 \times 1500$. 
A success is recorded if the relative recovery error is less than 0.001.
The success rate is averaged over 50 random realizations. 
 
Figure \ref{figure:DCT} shows success rates for the four algorithms with increasing 
factor $F$ from 2 to 8. Along with the increasing $F$, 
the success rates for the algorithms
decrease, though at different rates of decline. In all plots, TL1IT-s1 
is the best with the highest success rates.  
At $F=2$, both half thresholding and hard thresholding successfully
recover signal in the regime of small sparsity $k$. However when $F$ becomes
larger,  the half thresholding algorithm deteriorates sharply. 
Especially at $F = 8$, it lies almost flat.

\subsection{Signal Recovery in Noise}  
Let us consider recovering signal in noise based on the model $y = Ax + \varepsilon$, 
where $\varepsilon$ is drawn from independent Gaussian 
$\varepsilon \in \mathcal{N}(0, \sigma^2)$ with $\sigma = 0.01$. 
The non-zero entries of sparse vector $x$ are drawn from  
$\mathcal{N}(0,4)$. In order to recover signal with certain accuracy, the error $\varepsilon$ can not be too 
large. So in our test runs, we also limit the noise amplitude 
as $|\varepsilon|_{\infty} \leq 0.01$.  

\subsubsection*{Gaussian Sensing Matrix}
Here we use the same method in Part B to obtain Gaussian matrix $A$. Parameter $r$ and
sparsity $k$ are in the same set $\{ 0 , \ 0.2 , \ 0.4 , \ 0.5 \}$ and 
$\{ 5 , \ 8 , \ 11 , \ ... , \ 35 \}$.
Due to the presence of noise, it becomes harder to accurately recover the 
original signal $x$. So we 
tune down the requirement for a success to relative error 
$\frac{\|x^r - x\|}{\|x\|} \leq 10^{-2}$.

\begin{figure}
\begin{tabular}{ll}
\begin{minipage}[t]{0.4\linewidth}
\includegraphics[scale=0.325]{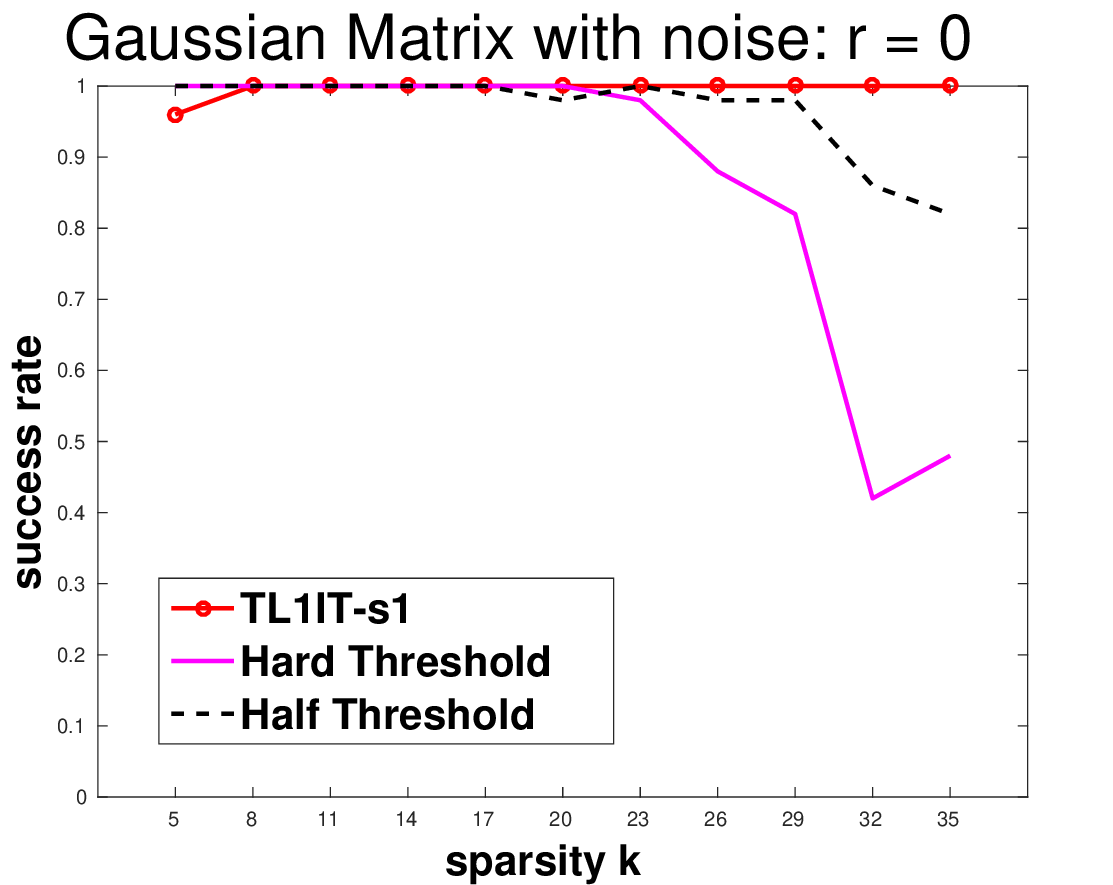}
\end{minipage}  &
\begin{minipage}[t]{0.4\linewidth}
\includegraphics[scale=0.325]{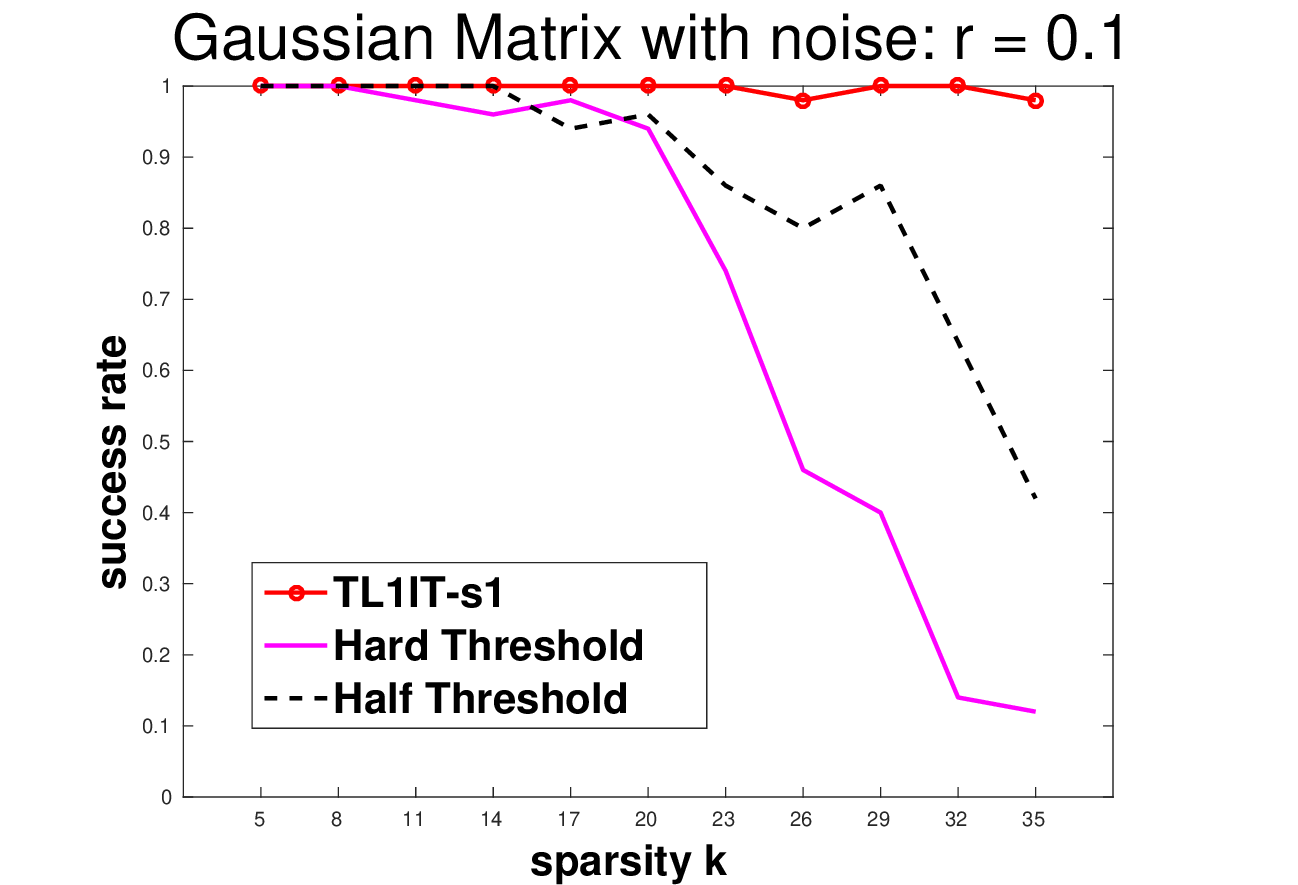}
\end{minipage}  \\
\begin{minipage}[t]{0.4\linewidth}
\includegraphics[scale=0.325]{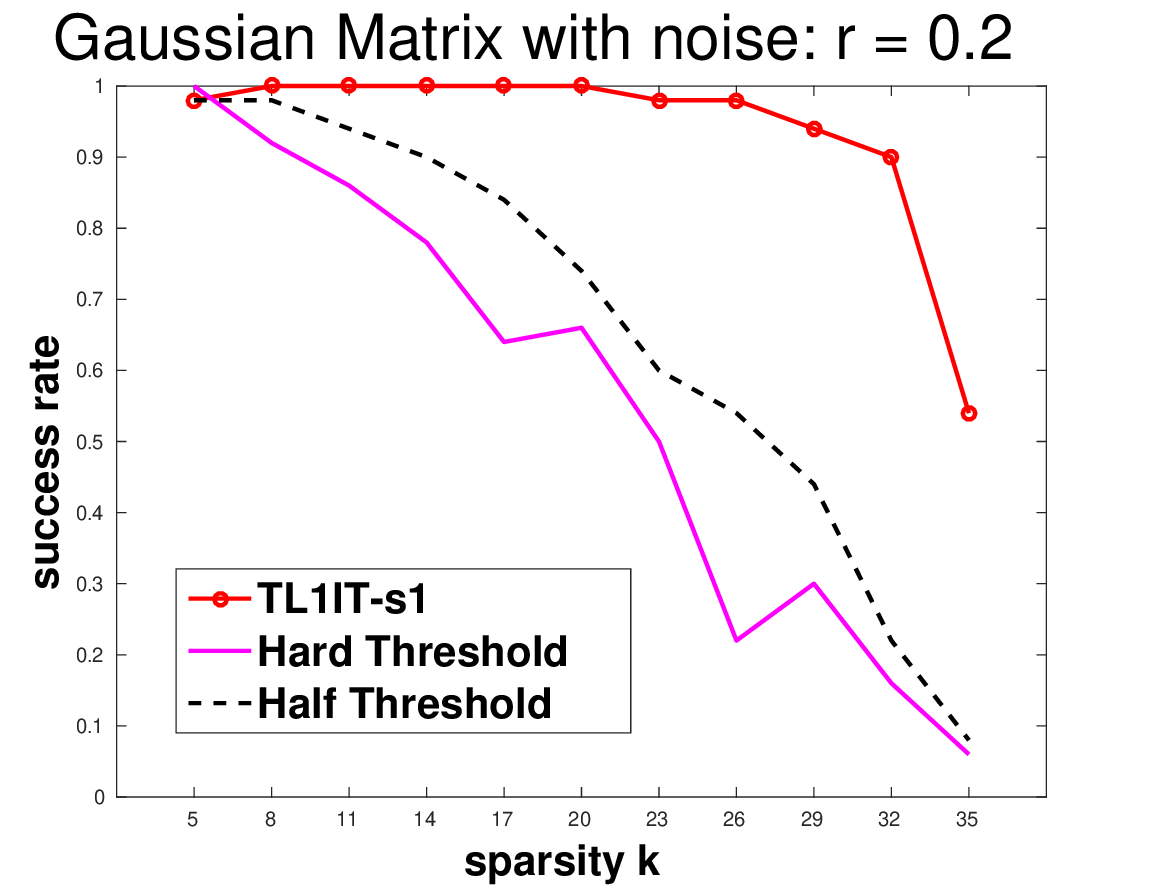}
\end{minipage}  &
\begin{minipage}[t]{0.4\linewidth}
\includegraphics[scale=0.325]{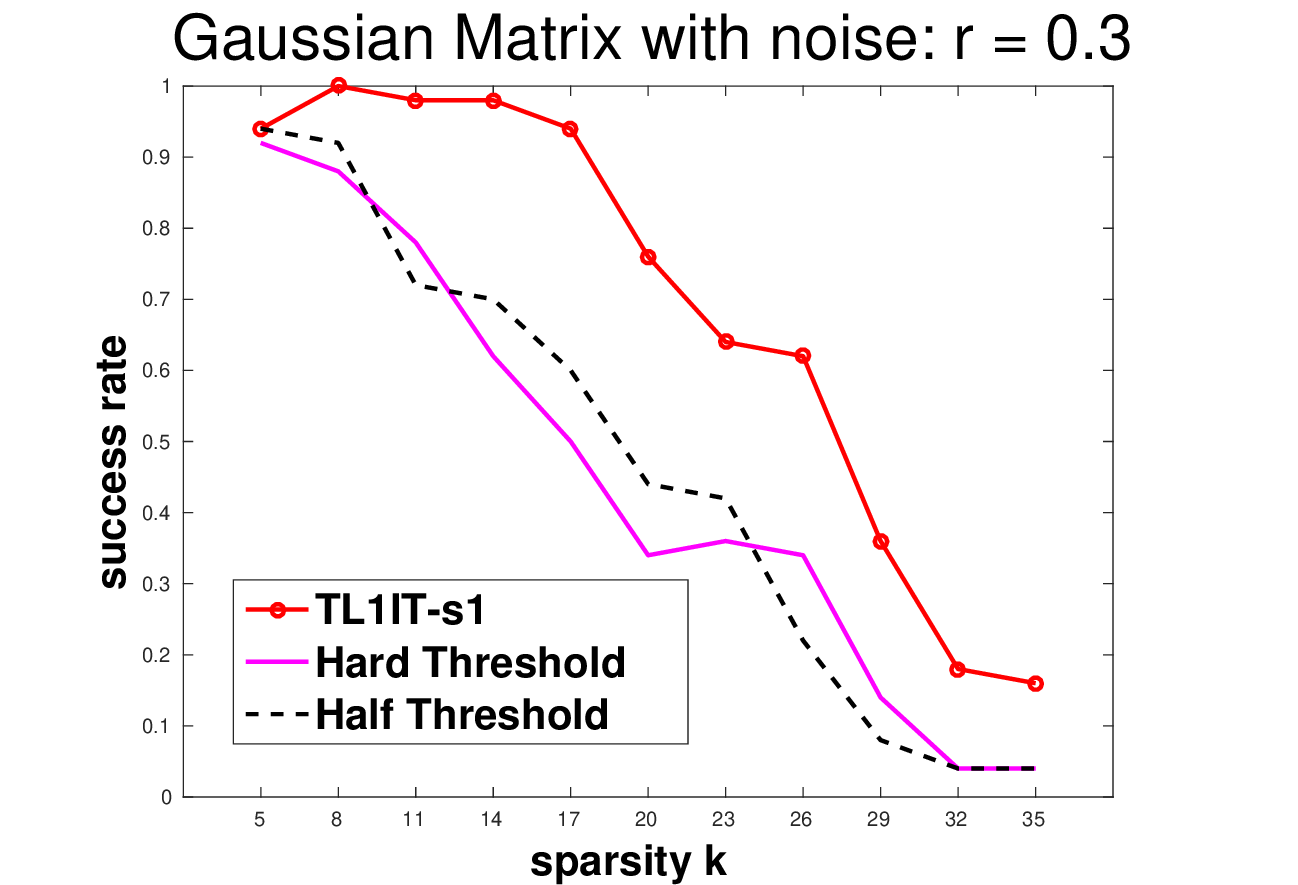}
\end{minipage} 
\end{tabular}
\caption{Algorithm comparison in success rates for $128\times 512$ 
Gaussian sensing matrices with additive noise at different coherence $r$.}
\label{figure:Gaussian noise}
\end{figure}    

The numerical results are shown in Figure \ref{figure:Gaussian noise}. In this experiment, 
TL1IT-s1 again has the best performance, with half thresholding algorithm the second. 
At $r = 0$, TL1IT-s1 scheme is robust and recovers signals successfully in almost 
all runs, which is the same case under both noisy and noiseless conditions.

\subsubsection*{Over-sampled DCT Sensing Matrix}
\FIG\ref{figure:DCT noise} shows results of three algorithms under the 
over-sampled DCT sensing matrices. Relative error of $0.01$ or under qualifies 
for a success. In this case, TL1IT-s1 is also the best numerical method, 
same as in the noise free tests. 
It degrades most slowly under high coherence sensing matrices ($F = 6,8$).  
  
\begin{figure}
\begin{tabular}{ll}
\begin{minipage}[t]{0.4\linewidth}
\includegraphics[scale=0.327]{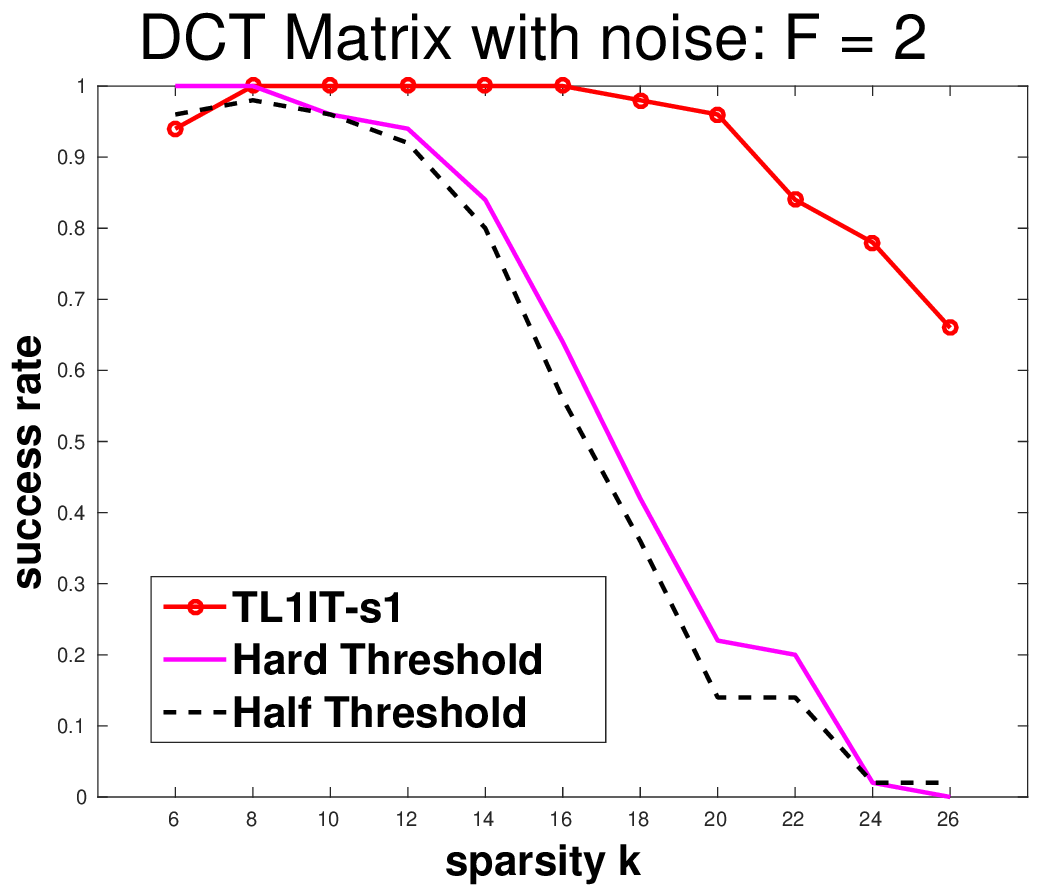}
\end{minipage}  &
\begin{minipage}[t]{0.4\linewidth}
\includegraphics[scale=0.325]{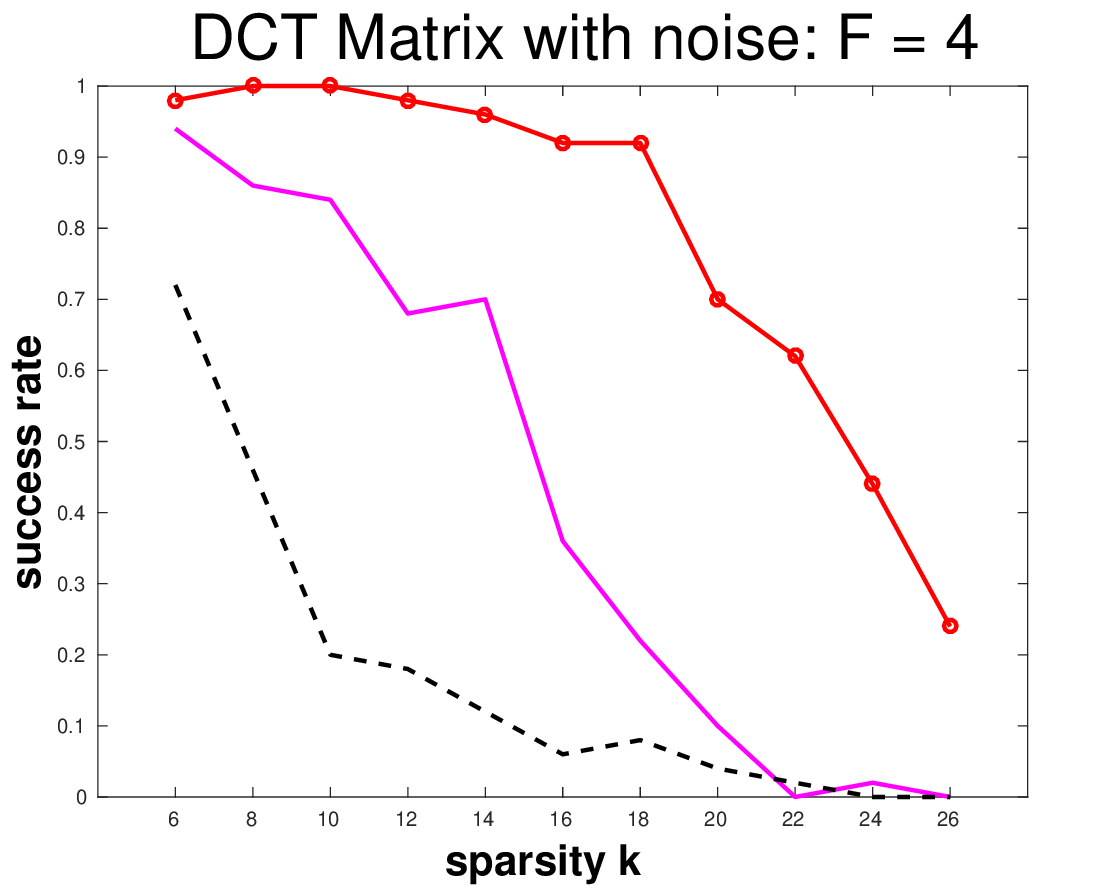}
\end{minipage}  \\
\begin{minipage}[t]{0.4\linewidth}
\includegraphics[scale=0.327]{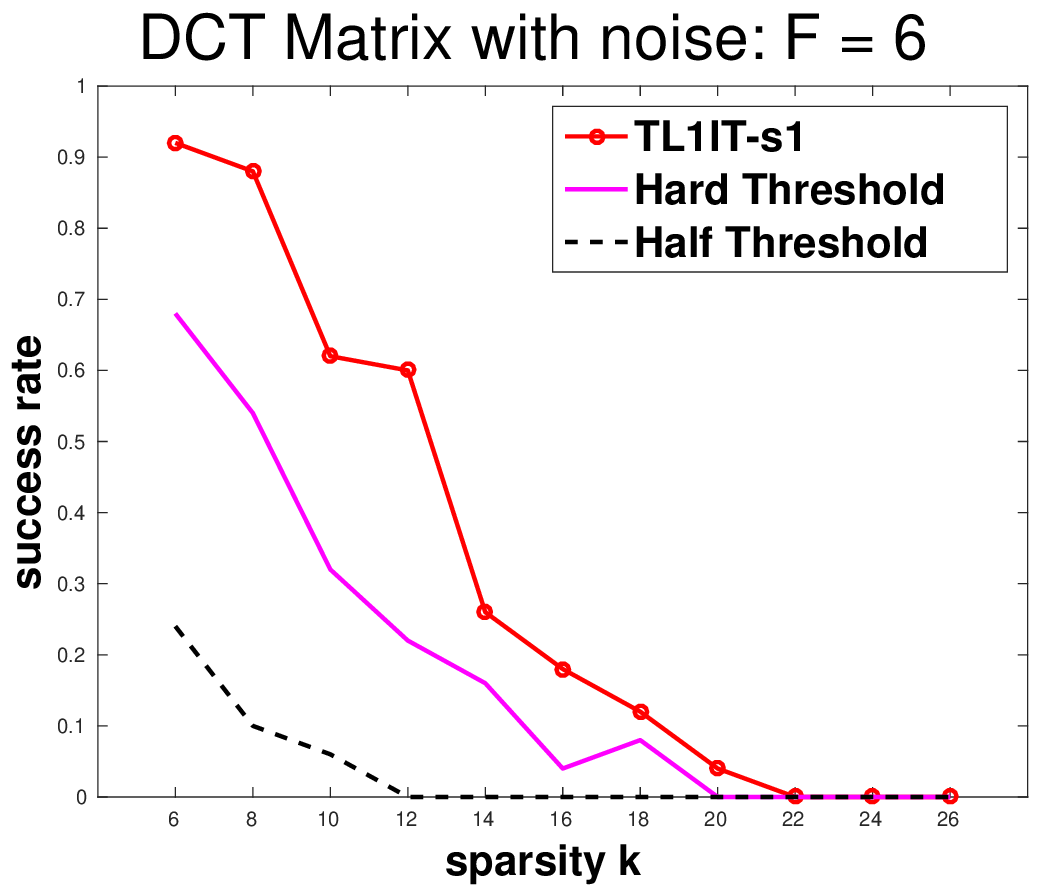}
\end{minipage}  &
\begin{minipage}[t]{0.4\linewidth}
\includegraphics[scale=0.325]{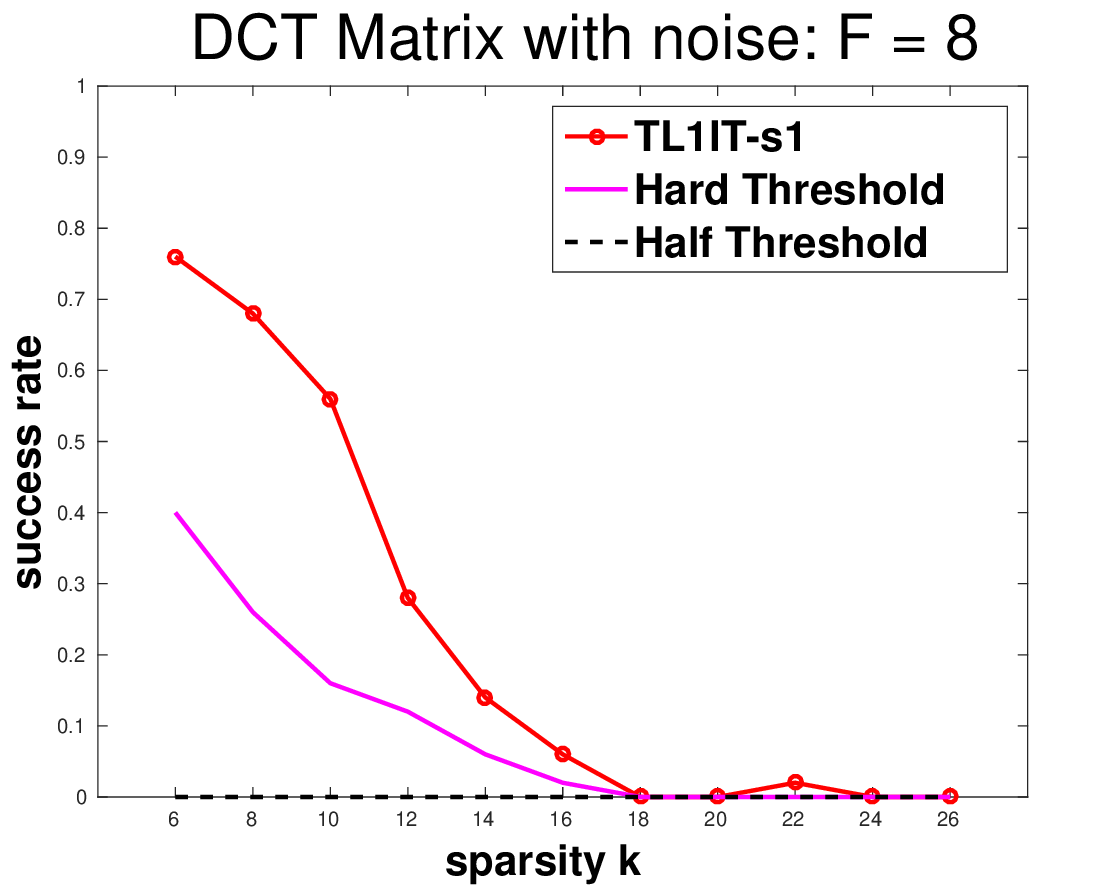}
\end{minipage} 
\end{tabular}
\caption{Algorithm comparison for over-sampled DCT matrices with additive noise: 
$M = 100$, $N = 1500$ at $F=2,4,6,8$.}
\label{figure:DCT noise}
\end{figure}

\subsection{Robustness under Sparsity Estimation}  
In the previous numerical experiments, the sparsity of the problem is known and 
used in all thresholding algorithms. However, in many applications, the sparsity 
of problem may be hard to know exactly. Instead, one may 
only have a rough estimate of the sparsity. 
How is the performance of the TL1IT-s1 when the exact sparsity $k$ is replaced by a rough 
estimate ? 

Here we perform simulations to verify the robustness of TL1IT-s1 algorithm 
with respect to sparsity estimation. Different from previous examples, 
Figure \ref{figure:robust} shows mean square error (MSE), 
instead of relative $l_2$ error. The sensing matrix $A$ is 
generated from Gaussian distribution with $r = 0$. Number of columns, $M$  
varies over several values, while the number of rows, $N$, is fixed at 512. 
In each experiment, we change the sparsity estimation for the algorithm from 60 to 240. 
The real sparsity is $k = 130$. This way, we test the robustness of the TL1IT algorithms 
under both underestimation and overestimation of sparsity.     

In Figure \ref{figure:robust}, we see that TL1IT-s1 scheme is 
robust with respect to sparsity estimation, especially for sparsity over-estimation.  
In other words, TL1IT scheme can withstand the estimation error if 
given enough measurements.

\begin{figure}
\begin{tabular}{lll}
\begin{minipage}[t]{0.30\linewidth}
\includegraphics[scale=0.245]{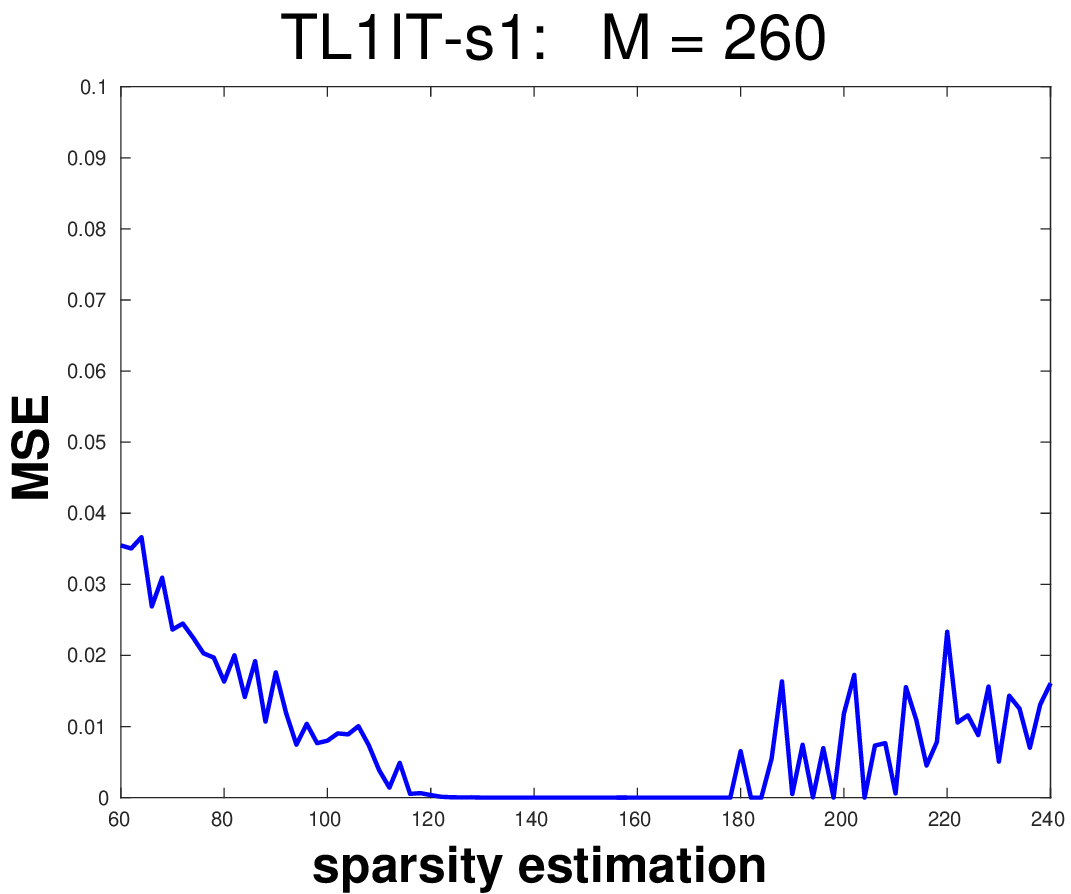}
\end{minipage}  &
\begin{minipage}[t]{0.30\linewidth}
\includegraphics[scale=0.245]{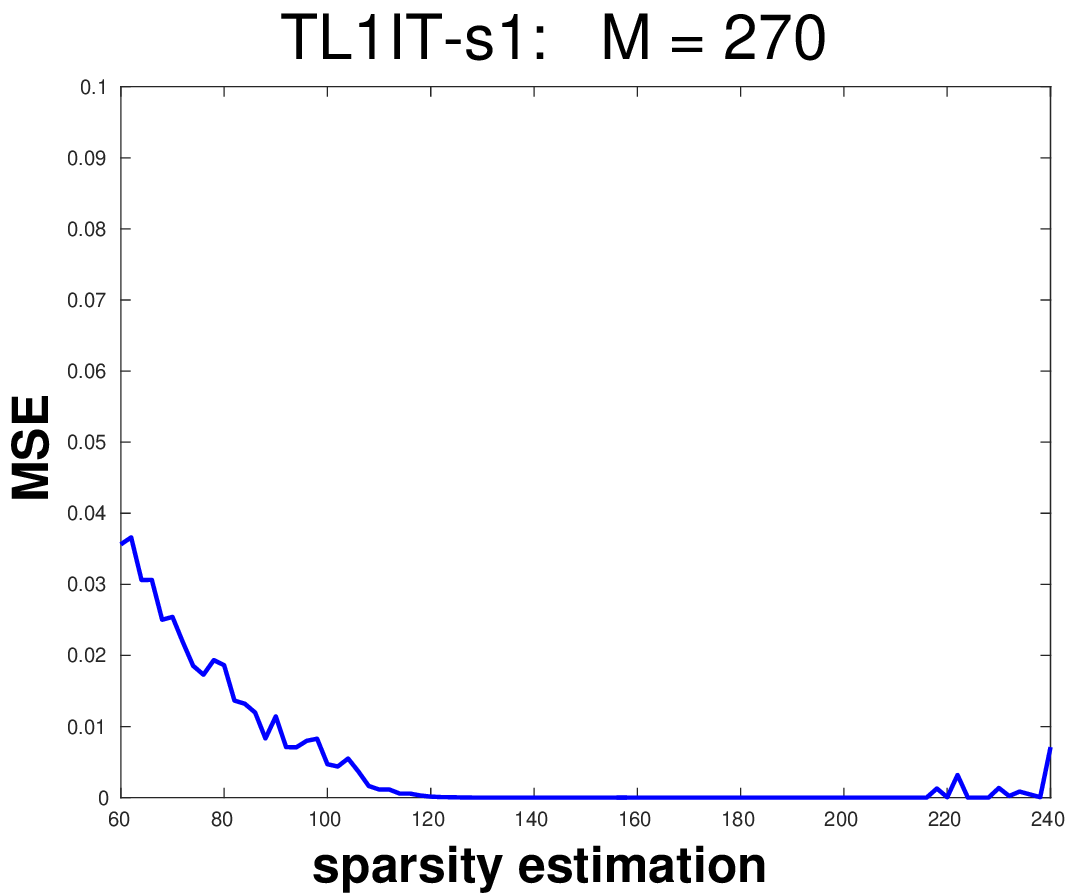}
\end{minipage}  &
\begin{minipage}[t]{0.30\linewidth}
\includegraphics[scale=0.245]{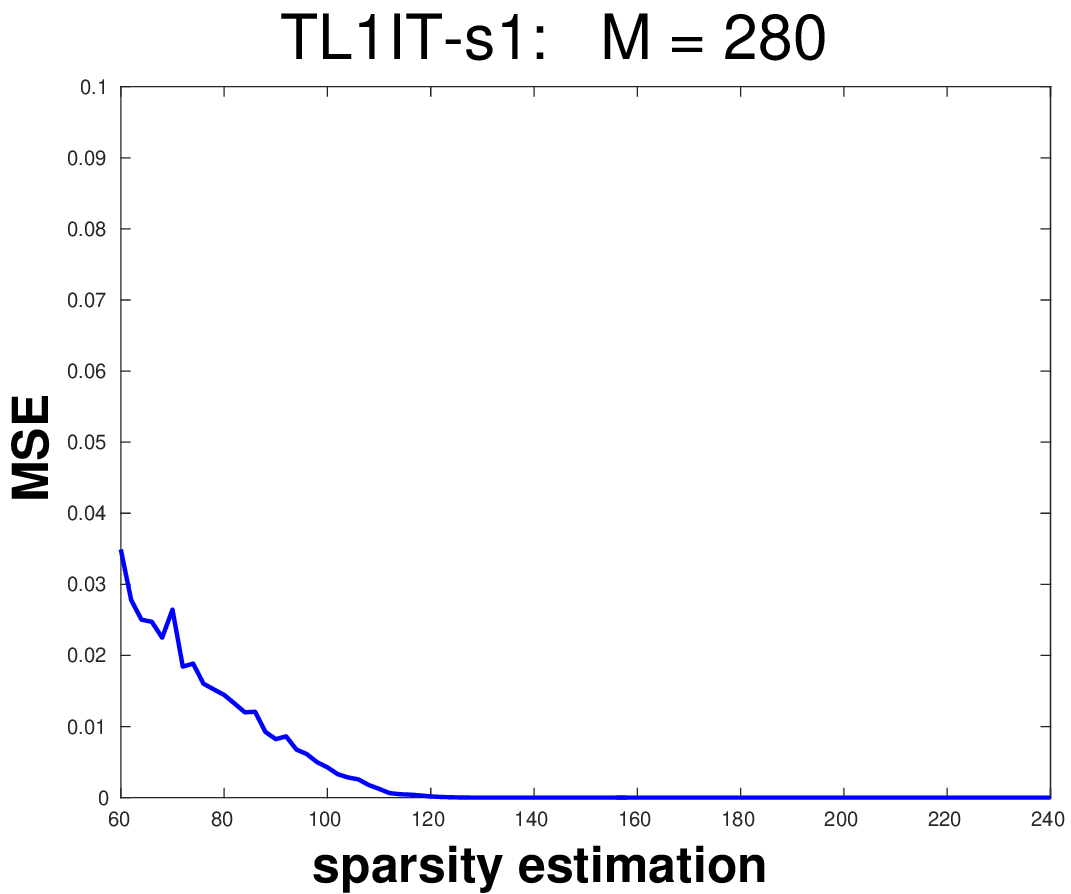}
\end{minipage}  
\end{tabular}
\caption{Robustness tests (mean square error vs. sparsity) for TL1IT-s1 thresholding algorithm 
under Gaussian sensing matrices: $r = 0, N = 512$ and number of measurements $M=260,270,280$. 
The real sparsity is fixed as $k=130$.}
\label{figure:robust}
\end{figure}     

\subsection{Comparison among TL1 Algorithms}
We have proposed three TL1 thresholding algorithms: DFA with fixed parameters, 
semi-adaptive algorithm -- TL1IT-s1 and adaptive algorithm --
 TL1IT-s2. Also in \cite{DCATL1}, we presented a TL1   
difference of convex function algorithm -- DCATL1.
Here we compare all four TL1 algorithms, 
under both Gaussian and Over-sampled DCT sensing matrices. 
For the fixed parameter DFA, we tested two thresholding schemes: 
DFA-s1 for continuous thresholding scheme under  
$\lambda \mu < a^2/2(a+1)$, and DFA-s2 for discontinuous thresholding scheme under   
$\lambda \mu > a^2/2(a+1)$.

\begin{figure}
\begin{tabular}{lr}
\begin{minipage}[t]{0.41 \linewidth}
\includegraphics[scale=0.338]{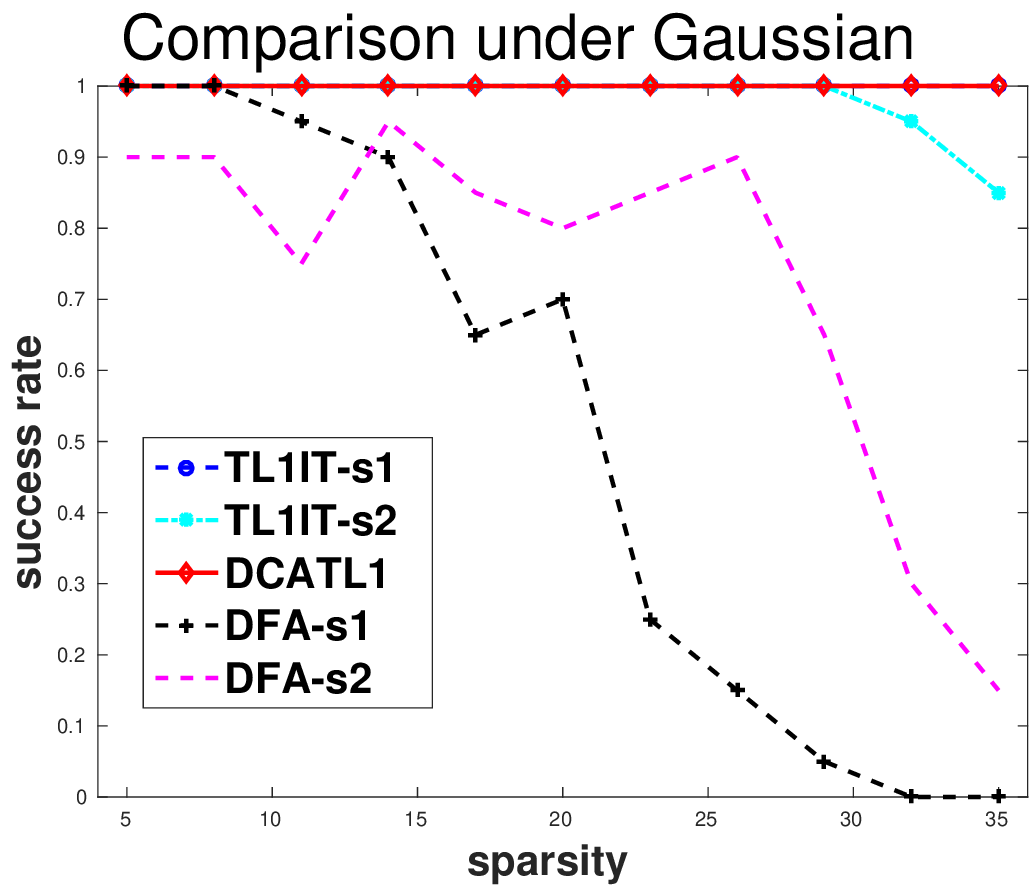}
\end{minipage}  &
\begin{minipage}[t]{0.41 \linewidth}
\includegraphics[scale=0.338]{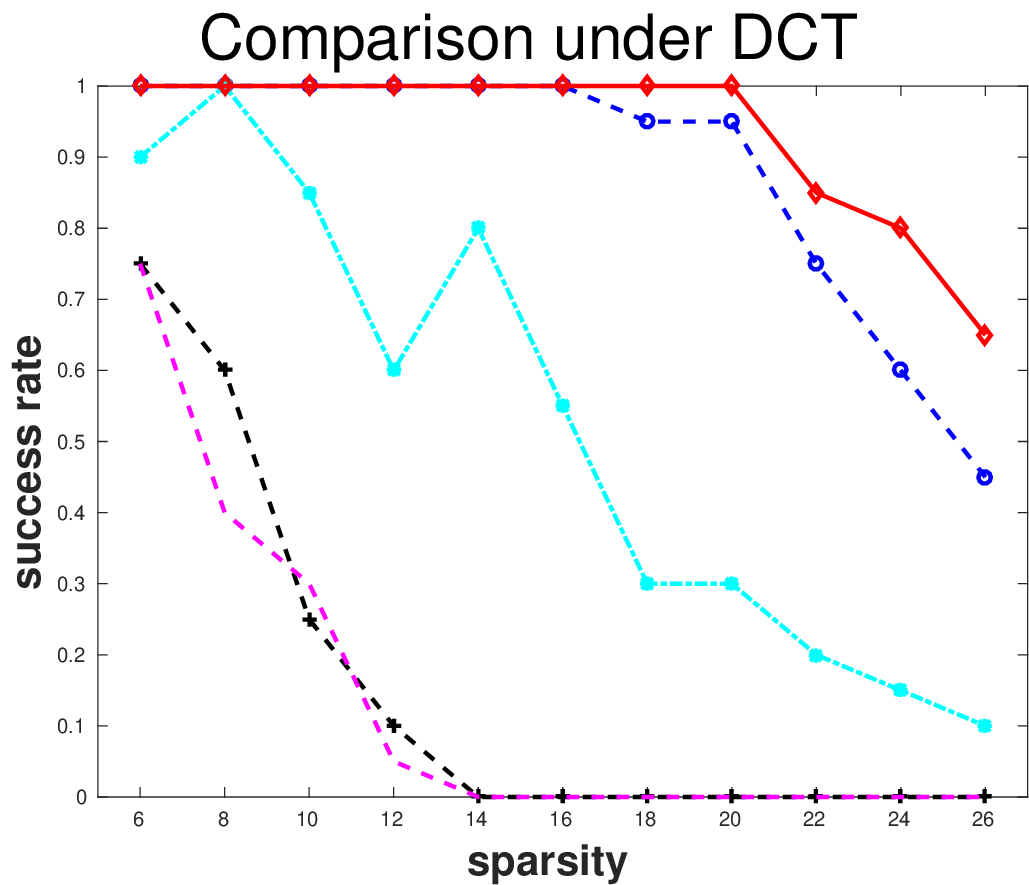}
\end{minipage}  
\end{tabular}
\caption{TL1 algorithms comparison. Y-axis is success rate from 20 random 
tests with accepted relative error $10^{-3}$. X-axis is sparsity value $k$.
Left: $128 \times 512$ Gaussian sensing matrices with sparsity $k = 5, \cdots, 35$. 
Right: $100 \times 1500$ Gaussian sensing matrices with sparsity $k = 6, \cdots, 26.$ 
}
\label{figure: TL1 comparison}
\end{figure}     

In the comparison experiments, we chose Gaussian matrices with 
covariance parameter $r=0$ and Over-sampled DCT matrices with $F=2$. The results
are showed in Figure \ref{figure: TL1 comparison}. 
Under Gaussian sensing matrices, 
DCATL1 and TL1IT-s1 achieved 100 \% success rate to
recover ground truth sparse vector, while TL1IT-s2 failed sometimes when 
sparsity is higher than 28. Also it is interesting to notice that DFA-s2 with
discontinuous thresholding scheme behaved better than DFA-s1, the continuous thresholding scheme. 
For over-sampled DCT sensing tests, DCATL1 is clearly the best among
all TL1 algorithms, with TL1IT-s1 the second. Also the performance of 
TL1IT-s2 declined sharply under this test, which is consistent with our previous
numerical experiments for thresholding algorithms. Due to this fact, 
we only showed TL1IT-s1 in the plots for comparison
with hard and half thresholding algorithms. 

The two adaptive TL1 thresholding algorithms are far ahead of 2 DFA
algorithms, which shows the advantages of adaptivity.
Although DCATL1 out-performed all TL1 thresholding algorithms in the above tests, 
it requires two nested iterations, and an inverse matrix operation,
which is costly for a large size sensing matrix. So 
for large scale CS applications, thresholding algorithms will have their 
advantages, including parallel implementations.
  
\section{Conclusion}  \label{section: conclusion}
We have studied compressed sensing problems with the transformed $l_1$
penalty function for the unconstrained regularization model. 
We established a precise thresholding representation theory with closed form thresholding formula,  
and proposed three iterative thresholding schemes.  
The TL1 thresholding schemes can be either continuous 
(as in soft-thresholding of $l_1$) or discontinuous (as in half-thresholding of
$l_{1/2}$),  depending on whether the parameters belong to the subcritical or
supercritical regime. Correspondingly, there are two parameter setting 
strategies for regularization parameter $\lambda$, 
when the $k$-sparsity problem is solved. 
A convergence theorem is proved for the fixed parameter
TL1 algorithm (DFA).
 
Numerical experiments showed that the semi-adaptive TL1It-s1 algorithm is 
the best performer for sparse signal recovery under sensing matrices with a   
broad range of coherence and under controlled measurement noise. 
TL1IT-s1 is also robust under sparsity estimation error.
 
In a future work, we plan to explore TL1 thresholding algorithms for imaging
science among other higher dimensional problems. 

%

\appendices

\section{Relations of three parameters: $t^*_1$, \ $t^*_2$ \ and \ $t^*_3$}
\begin{equation*}
\left\lbrace \begin{array}{l}
   \vspace{2mm}
   t^*_1 = \dfrac{3}{2^{2/3}} (\lambda a(a+1))^{1/3} -a; \\
   t^*_2 = \lambda \frac{a+1}{a}； \\
   t^*_3 = \sqrt{2\lambda (a+1)} - \frac{a}{2}.
\end{array} \right.
\end{equation*} 
In this appendix, we prove that 
$$t^*_1 \leq t^*_3 \leq t^*_2,$$ 
for all positive parameters $\lambda$ and $a$. 
Also when $ \lambda = \frac{a^2}{2(a+1)}$, they are equal to $\frac{a}{2}$.
\begin{enumerate}
\item $t^*_1 \leq t^*_3$.

Consider the following equivalent relations:
\begin{equation*}
\begin{array}{ll}
t^*_1 \leq t^*_3 
&  \Leftrightarrow \ \
       \dfrac{3}{2^{2/3}} (\lambda a(a+1))^{1/3} \leq \frac{a}{2} + \sqrt{2\lambda (a+1)} \\
&  \Leftrightarrow \ \
       0 \leq (\sqrt{2\lambda(a+1)})^{3} + \frac{a^3}{8} - \frac{15}{4}a(a+1)\lambda + 
\frac{3a^2}{4}\sqrt{2\lambda(a+1)} 
\end{array}
\end{equation*}
Denote $\beta = \sqrt{\lambda}$, then function 
$P(\lambda) = (\sqrt{2\lambda(a+1)})^{3} + \frac{a^3}{8} - \frac{15}{4}a(a+1)\lambda 
+ \frac{3a^2}{4}\sqrt{2\lambda(a+1)}$ can be rewriten as a cubic polynomial of $\beta$: 
$$\beta^3 (2(a+1))^{3/2} - \beta^2 \frac{15}{8}a(2(a+1)) + \beta \frac{3a^2}{4}\sqrt{2(a+1)}
+ \frac{a^3}{8}.$$ 
This polynomial can be factorized as
$$ (2(a+1))^{3/2} \left(\beta - \frac{a}{\sqrt{2(a+1)}}\right)^2 
                  \left(\beta + \frac{a}{8\sqrt{2(a+1)}}\right).$$ 
Thus for nonnegative parameter $\lambda = \beta^2$, it is always true that 
$P(\lambda) \geq 0$. Therefore, we have $t^*_1 \leq t^*_3$. They are equal to $a/2$ if and only if 
$\lambda = \frac{a^2}{2(a+1)}$.

\item $t^*_3 \leq t^*_2$. 

This is because
\begin{equation*}
\begin{array}{ll}
t^*_3 \leq t^*_2 
&  \Leftrightarrow \ \ \sqrt{2\lambda (a+1)} \leq \frac{a}{2} + \lambda \frac{a+1}{a} \\
&  \Leftrightarrow \ \ 2\lambda (a+1) \leq \frac{a^2}{4} + \lambda (a+1) 
                       + \lambda^2 \frac{(a+1)^2}{a^2} \\
&  \Leftrightarrow \ \ 0 \leq \left( \frac{a}{2} - \lambda \frac{a+1}{a} \right)^2.
\end{array}
\end{equation*}
So inequality $t^*_3 \leq t^*_2$ holds. 
Further, $t^*_3 = t^*_2 = a/2$ if and only if $\lambda = \frac{a^2}{2(a+1)}$. 
\end{enumerate}

\section{Formula of optimal value $y^*$ when $\lambda > \frac{a^2}{2(a+1)}$ \ and \ $t^*_1 < x < t^*_2$}
Define function $w(x) = x - g_{\lambda}(x) - \frac{a}{2}$, where
\begin{equation*}
g_{\lambda}(x) = sgn(x) \left\{ \frac{2}{3}(a+|x|)\, \cos(\frac{\varphi(x)}{3}) 
                                 -\frac{2a}{3} + \frac{|x|}{3} \right\}
\end{equation*}
with $ \varphi(x) = \arccos( 1 - \frac{27\lambda a(a+1)}{2(a+|x|)^3} )$.

\begin{enumerate}
\item
First, we need to check that $x = t^*_3$ indeed is a solution for equation
$w(x) = 0$. 

Since $\lambda > \frac{a^2}{2(a+1)}$, $t_3^* = \sqrt{2\lambda (a+1)} - \frac{a}{2} > 0$. 
Thus:
\begin{equation*}
\begin{array}{ll}
\vspace{2mm}
\cos(\varphi(t_3^*)) & = 1 - \dfrac{27\lambda a(a+1)}{2(a+t_3^*)^3} \\
                    & = 1 - \dfrac{27\lambda a(a+1)}{2( \frac{a}{2}+ \sqrt{2\lambda(a+1)} )^3}.
\end{array}
\end{equation*}
Further, by using the relation $\cos(\varphi) = 4\cos^3(\varphi/3) - 3\cos(\varphi/3)$ and
$ 0 \leq \varphi/3 \leq \frac{\pi}{3}$, 
we have 
\begin{equation*}
\cos\left(\frac{\varphi(t_3^*)}{3}\right) = 
\dfrac{ \sqrt{2\lambda(a+1)} - a/4 }{ a/2 + \sqrt{2\lambda(a+1)} }.
\end{equation*}
Plugging this formula into $g_{\lambda}(t^*_3)$ shows that 
$ g_{\lambda}(t^*_3) = \sqrt{2\lambda (a+1)} - a = t^*_3 - a/2$.
So $t^*_3$ is a root for function $w(t)$ and 
$t_3^* \in (t_1^*, t_2^*)$.

\item 
Second we prove that the function $w(x)$ changes sign at $x = t_3^*$.

Notice that according to Lemma \ref{lem: roots for poly} , $g_{\lambda}(x)$
is the largest root for cubic polynomial 
$P(t) = t(a+t)^2 - x(a+t)^2 + \lambda a(a+1)$, if $x > t^*_1$.  

Take $t=x$, we know $P(x) = \lambda a(a+1) > 0$. 
Let us consider the value of $P(x-a/2)$. It is easy to check that: 
$P(x-a/2) < 0 \Leftrightarrow x > t^*_3$.

\begin{enumerate}
\item $x \in (t^*_3, t^*_2)$. 

We will have $P(x-a/2) < 0$ and $P(x) > 0$. 
While also the largest solution of $P(t) = 0$ is $t = g_{\lambda}(x) < x$. 
Thus we are sure that $ g_{\lambda}(x) \in (x-a/2, x)$, and then 
$ x - g_{\lambda}(x) < a/2 \Rightarrow w(x) < 0.$ 
So the optimal value is $y^* = y_0 = g_{\lambda}(x)$.

\item $x \in (t^*_1, t^*_3)$. We have $P(x-a/2) > 0$ and $P(x) > 0$. 
Due to the proof of Lemma \ref{lem: roots for poly}, one possible situation is that
there are two roots $y_0$ and $y_1$ within interval $(x-a/2,x)$. But we can exclude this 
case. This is because, by formula (\ref{equ: three roots}), 
\begin{equation}
\begin{array}{ll}
\vspace{2mm}
y_0 - y_1 
& = \frac{2(a+x)}{3} \left\{ \cos(\varphi /3) - \cos(\varphi/3 + \pi/3) \right\} \\
\vspace{2mm}
& = \frac{2(a+x)}{3} \left\{ 2 \sin(\varphi/3 + \pi/6) \sin(\pi/6) \right\} \\
\vspace{2mm}
& = \frac{2(a+x)}{3} \sin(\varphi/3 + \pi/6).
\end{array}
\end{equation}
Here $ \varphi/3 \in [\pi/6, \pi/2]$. So $y_0 - y_1 \geq \frac{(a+x)}{3}$. 
Also we have $x > t^*_1 > a/2$ when $\lambda > \frac{a^2}{2(a+1)}$. 
Thus $$y_0 - y_1 > a/2,$$ which is in contradiction with the assumption that 
both $y_0$ and $y_1 \in (x-a/2,x)$. 
So there are no roots for $P(t) = 0$ in $(x -a/2,x)$. 
Then we know $y_0 = g_{\lambda}(x) < x-a/2$. That is to say, $w(x) > 0$, so the optimal 
value is $y^* = 0$. 
\end{enumerate}

\end{enumerate}

\section{Continuity of TL1 threshold function at $t^*_2$ when $\lambda \leq \frac{a^2}{2(a+1)}$}

Threshold operator $H_{\lambda, a}(\cdot)$ is defined as
\begin{equation*}
H_{\lambda, a}(x) = \left\{
\begin{array}{ll}
0, & \quad  \text{if} \ \ |x| \leq t; \\
g_{\lambda}(x), & \quad \text{if} \ \ |x| > t.
\end{array} \right.
\end{equation*}
When $\lambda \leq \frac{a^2}{2(a+1)}$, threshold value 
$t = t^*_2 = \lambda \frac{a+1}{a}$.

To prove continuity as shown in \FIG\ref{figure: threshold plot}, the satisfaction of
condition: $g_{\lambda}(t^*_2) = g_{\lambda}(-t^*_2) = 0$ is sufficient.

According to formula (\ref{func: g formula}), we substitute $x = \lambda \frac{a+1}{a}$ 
into function $\varphi(\cdot)$, then
\begin{equation*}
\begin{array}{ll}
\vspace{2mm}
\cos(\varphi) & = 1 - \dfrac{27\lambda a (a+1)}{ 2(a+ x)^3 } \\
              & = 1 - \dfrac{ 27\lambda a(a+1) }{2(a + \lambda \frac{a+1}{a})^3}.
\end{array}
\end{equation*} 
\begin{enumerate}
\item Firstly, consider $\lambda = \frac{a^2}{2(a+1)}$. Then $x = t^*_2 = \frac{a}{2}$, 
so $\varphi = \arccos(-1) = \pi$. Thus $\cos(\varphi/3) = \frac{1}{2}$. 
By taking this into function $g_{\lambda}$, it is easy to check that $g_{\lambda}(t^*_2) = 0$.

\item Then, suppose $\lambda < \frac{a^2}{2(a+1)}$. In this case, $x = t^*_2 > t^*_1$, 
so we have inequalities 
$$ 
-1 < d = \cos(\varphi) = 1 - \dfrac{ 27\lambda a(a+1) }{2(a + \lambda \frac{a+1}{a})^3} < 1.
$$
From here, we know $ \cos(\frac{\varphi}{3}) \in (\frac{1}{2},1)$. 

Due to triple angle formula: 
$4\cos^3(\frac{\varphi}{3}) - 3\cos(\frac{\varphi}{3}) = \cos(\varphi) = d$, 
let us define a cubic polynomial $c(t) = 4t^3 - 3t - d$. Then we have:
$c(-1) = -1 -d < 0$, $c(-1/2) = 1 -d > 0$, $c(1/2) = -1 -d < 0$ and $c(1) = 1 -d > 0$.
So there exist three real roots for $c(t)$, and only one root is located in $(1/2,1)$.

Further, we can check that $t^* = \dfrac{a - \frac{\lambda(a+1)}{2a}}{a + \frac{\lambda(a+1)}{a}}$
is a root of $c(t) = 0$ and also under the condition $\lambda < \frac{a^2}{2(a+1)}$, 
$\frac{1}{2} < t^* < 1$. From above discussion and triple angle formula, we can figure out 
that $\cos(\frac{\varphi}{3}) = \dfrac{a - \frac{\lambda(a+1)}{2a}}{a + \frac{\lambda(a+1)}{a}}$. 
Further, it is easy to check that 
$g_{\lambda}(t^*_2) = 0$.

\end{enumerate}


\newpage
\begin{thebibliography}{}
\bibitem{hard-threshold-blumensath2008iterative}
T. Blumensath, M. Davies.
\newblock Iterative thresholding for sparse approximations.
\newblock {\em Journal of Fourier Analysis and Applications}, 14(5-6):629-654,
  2008.

\bibitem{hard-sparsify-blumensath2012accelerated} T. Blumensath.
\newblock Accelerated iterative hard thresholding.
\newblock {\em Signal Processing}, 92(3):752-756, 2012.
  
\bibitem{candes2005decoding}E. Cand\`es, T. Tao,
{\em Decoding by linear programming}, IEEE Trans. Info. Theory, 51(12):4203-4215, 2005.

\bibitem{candes2006stable}E. Cand\`es, J. Romberg, T. Tao,
{\em Stable signal recovery from incomplete and inaccurate measurements},
Comm. Pure Applied Mathematics, 59(8):1207-1223, 2006.

\bibitem{ReweightedL1} E. Cand\`es, MB. Wakin and SP. Boyd, 
{\em Enhancing sparsity by reweighted $\ell_1$ minimization}, 
Journal of Fourier analysis and applications 14.5-6 (2008): 877-905.

\bibitem{Super:candes2013mini}E. Cand\`es, C. Fernandez-Granda, 
{\em Super-resolution from noisy data}, Journal of Fourier Analysis and Applications, 19(6):1229-1254, 2013.

\bibitem{xian-half-cao2013fast-image}W. Cao, J. Sun, and Z. Xu.
\newblock Fast image deconvolution using closed-form thresholding formulas of
  regularization.
\newblock {\em Journal of Visual Communication and Image Representation},
  24(1):31-41, 2013.

\bibitem{soft-threshold-lp-daubechies2004iterative}
I. Daubechies, M. Defrise, and C. De~Mol.
\newblock An iterative thresholding algorithm for linear inverse problems with
  a sparsity constraint.
\newblock {\em Communications on pure and applied mathematics},
  57(11):1413-1457, 2004.
  
\bibitem{Don_95}D. Donoho, 
{\em Denoising by soft-thresholding}, IEEE Trans. Info. Theory, 41(3), pp. 613--627, 1995.

\bibitem{Don_06}D. Donoho,
{\em Compressed sensing}, IEEE Trans. Info. Theory, 52(4), 1289-1306, 2006.

\bibitem{ELX} E. Esser, Y. Lou and J. Xin,
{\em A Method for Finding Structured Sparse Solutions to Non-negative Least Squares
Problems with Applications}, SIAM J. Imaging Sciences, 6(2013), pp. 2010-2046.


\bibitem{SCAD} J. Fan, and R. Li,
{\em Variable selection via nonconcave penalized likelihood and its 
oracle properties}, Journal of the American Statistical Association 96(456), 
pp. 1348-1360, 2001.

\bibitem{Fann_12}A. Fannjiang, W. Liao,
{\em Coherence Pattern-Guided Compressive Sensing with Unresolved Grids},
SIAM J. Imaging Sciences, Vol. 5, No. 1, pp. 179-202, 2012.

\bibitem{l1-l2-lou2014computing}Y. Lou, P. Yin, Q. He, and J. Xin, 
{\em Computing Sparse Representation in a Highly Coherent Dictionary Based on Difference of L1 and L2,} 
 J. Sci. Computing, 1 (2015), pp. 178-196.

\bibitem{transformed-l1}J. Lv, and Y. Fan, 
{\em A unified approach to model selection and sparse recovery using regularized least squares,}
Annals of Statistics, 37(6A), pp. 3498-3528, September 2009.

\bibitem{SparseNet}R. Mazumder, J. Friedman, and T. Hastie,
{\em SparseNet: Coordinate descent with nonconvex penalties},
Journal of the American Statistical Association, 106(495), pp. 1125-1138, 2011.

\bibitem{Proximal}N. Parikh and SP. Boyd, 
{\em Proximal Algorithms}, 
Foundations and Trends in optimization 1.3 (2014): 127-239.

\bibitem{xian-half}Z. Xu, X. Chang, F. Xu, and H. Zhang, 
{\em $L_{1/2}$ regularization: A thresholding representation theory and a fast solver}, 
IEEE Transactions on Neural Networks and Learning Systems, 23(7):1013-1027, 2012.
  
\bibitem{yall1}J. Yang and Y. Zhang,
{\em Alternating direction algorithms for $l_1$ problems in compressive sensing}, 
 SIAM Journal on Scientific Computing, 33(1):250-278, 2011.
 
\bibitem{l1-l2-yinminimization}P. Yin, Y. Lou, Q. He, and J. Xin, 
{\em Minimization of $L_{1-2}$ for compressed sensing}, 
SIAM J. Sci. Computing, 37(1), pp. A536-A563, 2015. 
 
\bibitem{Bregman:yin2008}W. Yin, S. Osher, D. Goldfarb, and J. Darbon,
{\em Bregman iterative algorithms for $l_1$-minimization 
with applications to compressed sensing},
 SIAM Journal on Imaging Sciences, 1(1):143-168, 2008. 
 
\bibitem{YO_13}W. Yin and S. Osher, 
{\em Error Forgetting of Bregman Iteration}, J. Sci. Computing, 54(2), pp. 684--695, 2013.

\bibitem{xian-half-theory-2014}J. Zeng, S. Lin, Y. Wang, and Z. Xu, 
{\em $L_{1/2}$ regularization: Convergence of iterative half thresholding algorithm},
IEEE Transactions on Signal Processing 62.9 (2014): 2317-2329. 

\bibitem{MC+} C. Zhang,
{\em Nearly unbiased variable selection under minimax concave penalty},
The Annals of statistics (2010): 894-942.

\bibitem{DCATL1}S. Zhang and J. Xin,
{\em Minimization of transformed $L_1$ penalty: theory, difference of convex
function algorithm, and robust application in compressed sensing},
arXiv:1411.5735, 2014; CAM Report 14-68, UCLA.

\bibitem{TS1}S. Zhang, P. Yin and J. Xin,
{\em Transformed Schatten-1 Iterative Thresholding Algorithms for Low Rank Matrix Completion},
arXiv:1506.04444, 2015.

\bibitem{CL1} T. Zhang,
{\em Multi-stage convex relaxation for learning with sparse regularization},
Advances in Neural Information Processing Systems, pp. 1929-1936, 2009.



\end{thebibliography}
\end{document}